\def\showall{0}
\def\useieeelayout{0}

\newcommand{\inACC}[1]{\if\useieeelayout1{#1}\fi\if\showall1{\color{green!50!black}In ACC: #1}\fi}

\newcommand{\inArxiv}[1]{\if\useieeelayout0{#1}\else\if\showall1{\color{blue}In ArXiV: #1}\fi\fi}

\if\useieeelayout1
\documentclass[letterpaper, 10 pt, conference]{ieeeconf}  
\IEEEoverridecommandlockouts                              %
\usepackage{commands}

\title{\LARGE \bf
Quadrotor Control on $\mathrm{SU(2)}\times \Real^3$ with SLAM Integration$^*$
}

\author{Marcus Greiff$^{1}$, Patrik Persson$^{2}$, Zhiyong Sun$^{3}$, Karl \AA str\"om$^{2}$, and Anders Robertsson$^{1}$.%
\thanks{*The research leading to these results has  received funding from ELLIIT, and the Swedish Science Foundation (SSF) project “Semantic mapping and visual navigation for smart robots” (RIT15-0038).}%
\thanks{$^{1}$Department of Automatic Control, Lund University, Sweden, email: \texttt{[marcus.greiff,anders.robertsson]@control.lth.se}}
\thanks{$^{2}$Center for Mathematical Sciences, Lund University, Sweden, email:\texttt{[patrik.persson,kalle.astrom]@maths.lth.se}}
\thanks{$^{2}$Department of Electrical Engineering, Eindhoven University of Technology, the Netherlands, email:\texttt{z.sun@tue.nl}}
}

\else
\documentclass{article}
\usepackage{arxivtest,graphics,epsfig,times}
\usepackage{amsmath,amssymb,amsthm}
\usepackage{commands}
\title{Quadrotor Control on $\mathrm{SU(2)}\times \Real^3$ with SLAM Integration}
\author{Marcus Greiff, Patrik Persson, Zhiyong Sun, Karl \AA str\"om, and Anders Robertsson}

\fi

\usepackage{tikz}
\usetikzlibrary{shapes,arrows,positioning,fit,calc}

\begin{document}

    \maketitle

\if\useieeelayout1
\maketitle
\thispagestyle{empty}
\pagestyle{empty}
\fi

\begin{abstract}

\change{We present a trajectory tracking controller for a quadrotor unmanned aerial vehicle (UAV) configured on $\SUT\times \Real^3$, and relate this result to a family of geometric tracking controllers on $\SOT\times \Real^3$. The theoretical results are complemented by simulation examples, and the controller is subsequently implemented in practice and integrated with a simultaneous localization and mapping (SLAM) system through an extended Kalman filter (EKF). This facilitates the operation of the UAV without external motion capture systems, and we demonstrate that the proposed control system can be used for inventorying tasks in a supermarket environment without external positioning systems.}

\end{abstract}

\section{INTRODUCTION}
The UAV is quickly becoming a ubiquitous tool in modern society. There already exist commercially available products that are capable of autonomous flight through narrow forest paths while filming downhill mountain-bikers~\cite{skydio2021}. These robots have the potential to completely disrupt and revolutionize transportation and logistics~\cite{fernandez2019towards}. Due to the high margin pressure associated with these sectors, initial applications are likely to be found in inventorying and data collection using small UAVs, where the cost can be kept low and the solutions can be implemented under contemporary legislation. This motivates the development of control systems similar to~\cite{skydio2021} that can operate without external motion capture systems and perform simple inventorying tasks.

For this purpose, we present a nonlinear tracking controller for the UAV dynamics, which yields uniform local exponential stability (ULES) properties on a smaller domain of attraction, and asymptotic attractiveness on a much larger domain. Similar methods have been developed in the quaternion formalism~\cite{fresk2013full,brescianini2013nonlinear,brescianini2018tilt}, but the approach taken in this paper extends the results in~\cite{greiff2021similarities} using a special distance on $\SUT$. This permits a stability proof analogous to that of the geometric tracking controller on $\SOT$ in~\cite{lee2010geometric,goodarzi2013geometric}, which allows further generalization to the robust and globally stabilizing controllers in~\cite{lee2013nonlinear,lee2015global}. When compared to controllers based on the model predictive control (MPC) method, such as the learning-based MPC in~\cite{kaufmann2020deep} or the perception-aware MPC in~\cite{falanga2018pampc}, the geometric controllers are capable of similar feats of agility while requiring significantly less computational resources. This becomes particularly relevant in the context of small UAVs, which operate under computational constraints, where a large part of the processing power needs to be allocated to the on-board sensor fusion.

As the proposed controller relies on full-state information, a specialized system for simultaneous localization and mapping (SLAM) is implemented to generate real-time pose (position and attitude) estimates of the UAV  based on video feed streamed from a camera mounted on the UAV. For this purpose, a solution is implemented with ORB features \cite{orb} and the pre-integration proposed in \cite{preintegration}, permitting real-time estimation of the UAV state conditioned on the inertial measurement unit (IMU) measurements and pose estimates from the SLAM system. We emphasize that this solution enables safe autonomous flights without external motion capture systems, and demonstrate the theory with experiments using the Crazyflie 2.0 UAV~\cite{crazyflie2021twopointo} (see Fig.~\ref{fig:UAVconfigurations}).

\begin{figure}[h!]
\centering
\includegraphics[width=0.47\textwidth]{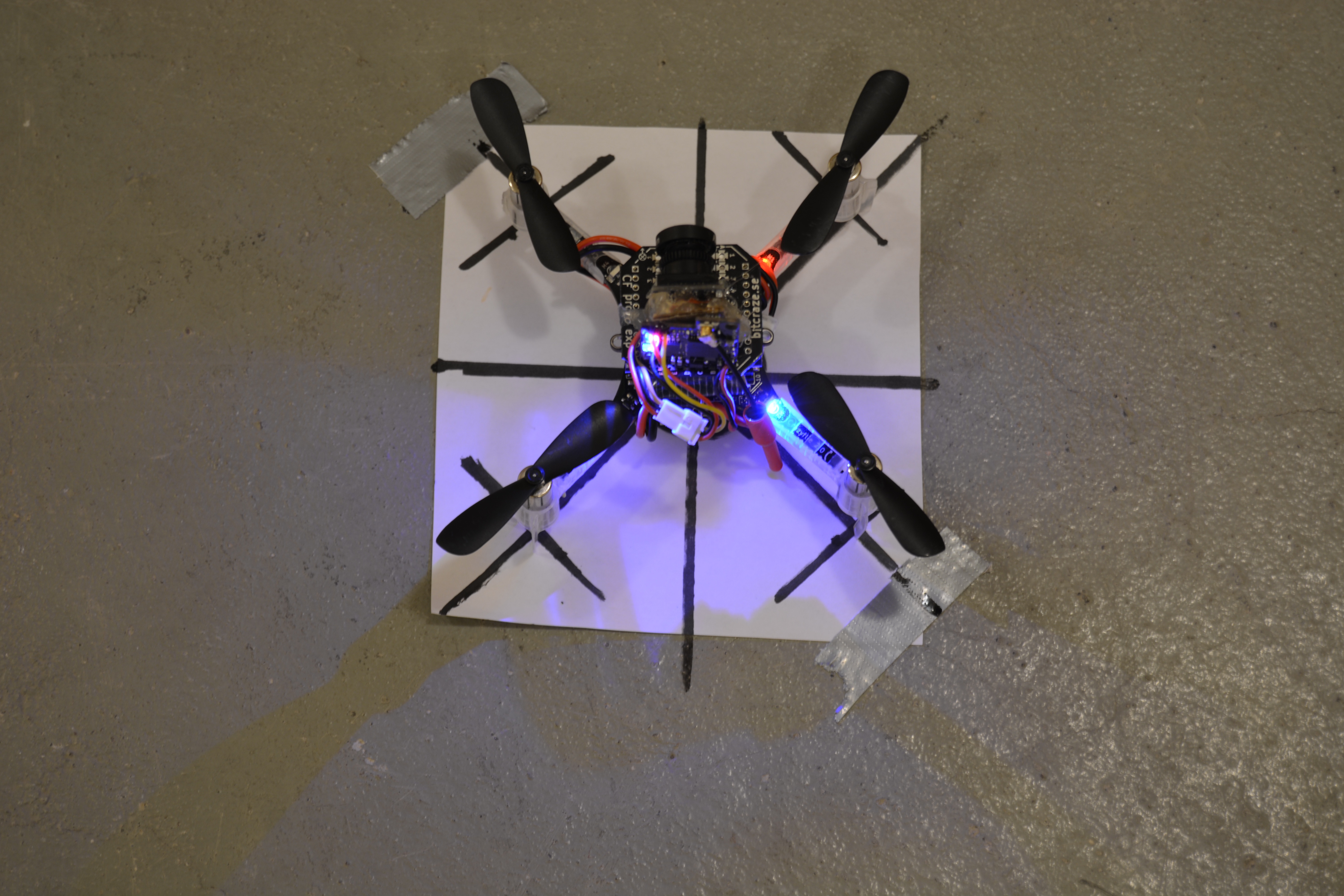}%
\hspace{3pt}\includegraphics[width=0.47\textwidth]{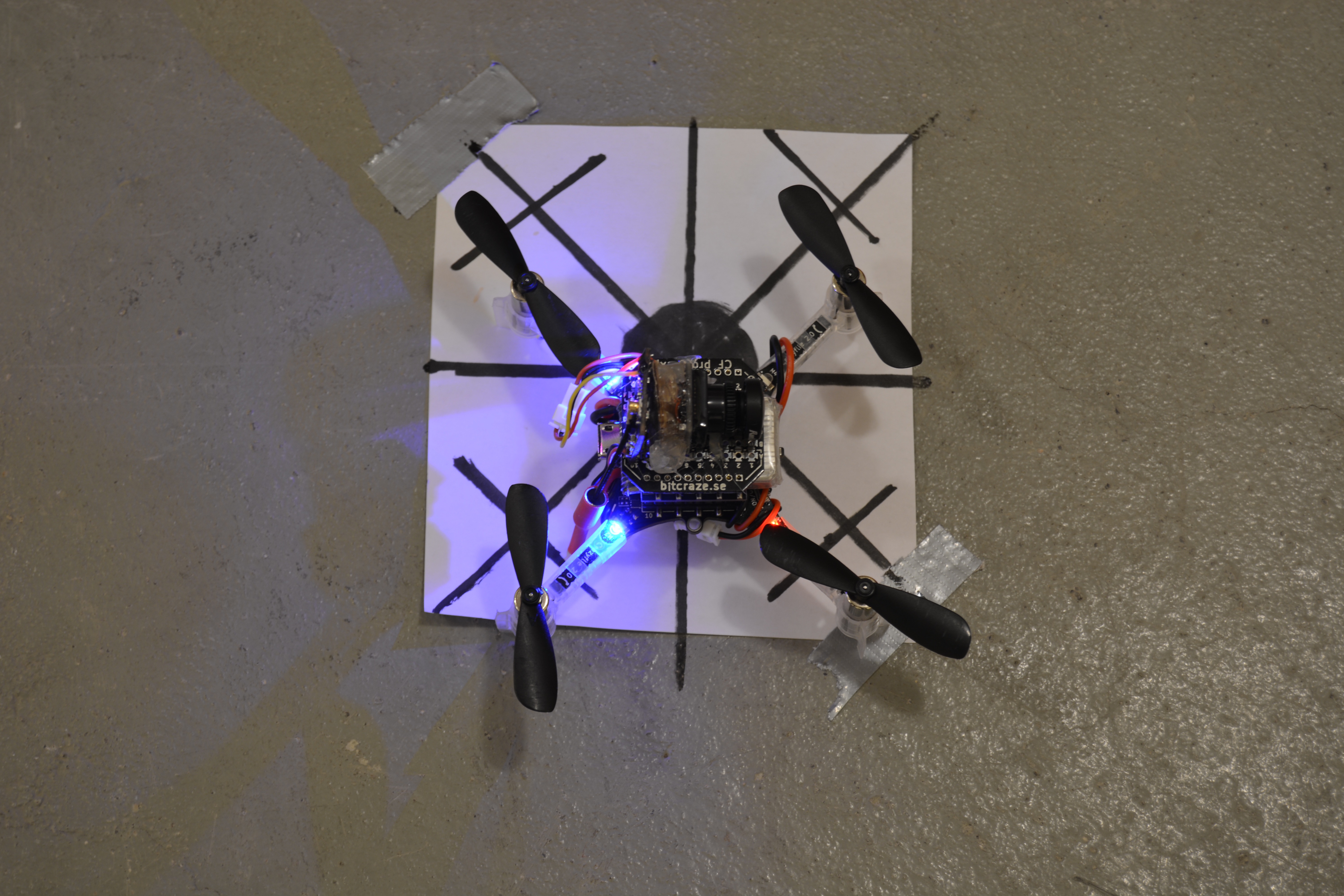}
\caption{The Crazyflie 2.0 used in the experiments with a camera attached. \textit{Left}: Initial configuration at $t=\tz$. \textit{Right}: Terminal configuration at $t=t_f$.}
\label{fig:UAVconfigurations}
\end{figure}


\inArxiv{\newpage}\subsection{Contributions}
The contributions of this paper are threefold, and can be summarized as follows: 
\begin{itemize}
    \item A  tracking geometric controller for UAV dynamics configured on $\SUT\times \Real^3$ is proposed, using the approach in~\cite{lee2010control} based on the attitude control result in~\cite{greiff2021similarities}. The controller differs to that in~\cite{lee2010control} in several ways that have meaningful consequences, and its stability proof further motivates the use of the results presented in~\cite{brescianini2018tilt}.
    \item A real-time compatible stack for SLAM is presented. The system tracks and matches ORB features \cite{orb} from a continuous video feed, utilizing the pre-integration proposed in \cite{preintegration} to leverage the sampled IMU-data.
    \item A demonstration of the controller in real-time is given, fusing estimates from the SLAM system and with other sensory data in an on-board extended Kalman filter (EKF), rendering the resulting control system capable of safely executing complex maneuvers that are relevant to supermarket inventorying.
\end{itemize}

\subsection{Notation} 
Matrices and vectors are written in bold font, with entries of a vector $\uvec$ as $u_i$, and the entries of a matrix \change{$\A$ are denoted by $[\A]_{ij}$}. The smallest eigenvalue of a real symmetric matrix \change{$\A\in\mathbb{R}^{n\times n}$} is denoted by \change{$\mineig(\A)$}, and its largest by \change{$\maxeig(\A)$}. The two-norm of $\uvec\in\Real^{n}$ is denoted by $\|\uvec \|=\sqrt{\uvec\Tr\uvec}$, and $\|\uvec \|_{\M} =\sqrt{\uvec\Tr\M\uvec}$ with a positive definite real matrix ${\M} \in\mathbb{R}^{n\times n}$. The trace of \change{$\A\in\mathbb{C}^{n\times n}$ is written $\Trace(\A)=\sum_{i = 1}^n A_{ii}$}, and $\Ucal(\Omega)$ denotes a uniform distribution over $\Omega$, and implies sampling every element of $\Omega$ with equal probability. Let $\Sbf:\Real^3\mapsto \Real^{3\times 3}$ such that for any $\avec, \bvec\in\Real^3$, $\Sbf(\avec)\bvec = \avec\times \bvec$\inACC{.}\inArxiv{, where
\begin{equation}
\Sbf(\avec) = \begin{bmatrix}
          0&  -a_3&  a_2\\
        a_3&     0& -a_1\\
       -a_2&   a_1&    0
\end{bmatrix}.
\end{equation}}

\subsection{Structure}
The mathematical preliminaries are given in Section~\ref{sec:prelim}, introducing a result on attitude control and identities pertaining to {\SUT}. This is followed by a presentation of the problem formulation and the main result in Section~\ref{sec:control}, where attitude controller is used to derive a nonlinear tracking feedback law for a UAV configured on $\SUT\times \Real^3$. The SLAM system and its integration are discussed in Section~\ref{sec:estimation}, followed by simulation and experimental results in Section~\ref{sec:numerical}. Finally, the conclusion and outlook in Section~\ref{sec:conclusion} close the paper.

\newpage\section{PRELIMINARIES}\label{sec:prelim}
In the mathematical preliminaries, we start by defining the configuration manifolds, and relate elements of $\SOT$ to elements $\SUT$ through a carefully constructed embedding.
\begin{definition}
Let $\SOT = \{\R\in\mathbb{R}^{3\times 3}\;|\; \R\Tr\R=\I, \;\det(\R)=1\}$, with an associated Lie algebra $\mathfrak{so}(3) = \{\Lbf_1\omega_1+\Lbf_2\omega_2+\Lbf_3\omega_3\in\mathbb{R}^{3\times 3}\;|\;\omegabf\in\mathbb{R}^3\}$ spanned by\label{def:SOT}
\begin{equation*}
\Lbf_1
\hspace{-2pt}=\hspace{-2pt}
\begin{bmatrix}
0&0&0\\
0&0&-1\\
0&1&0\\
\end{bmatrix}
\hspace{-2pt},\;
\Lbf_2
\hspace{-2pt}=\hspace{-2pt}
\begin{bmatrix}
0&0&1\\
0&0&0\\
-1&0&0\\
\end{bmatrix}
\hspace{-2pt},\;
\Lbf_3
\hspace{-2pt}=\hspace{-2pt}
\begin{bmatrix}
0&-1&0\\
1&0&0\\
0&0&0\\
\end{bmatrix}
\hspace{-3pt}.
\end{equation*}
\end{definition}
\vspace{4pt}

\begin{definition}\label{def:SUT}
Let $\SUT = \{\Xbf\in\mathbb{C}^{2\times 2}\;|\; \Xbf\Hr\Xbf=\I, \;\det(\Xbf)=1\}$, with an associated Lie algebra $\mathfrak{su}(2) = \{\Lbf_1\omega_1+\Lbf_2\omega_2+\Lbf_3\omega_3\in\mathbb{C}^{2\times 2}\;|\;\omegabf\in\mathbb{R}^3\}$ spanned by\label{def:SUT}
\begin{equation*}
\Lbf_1 =
\begin{bmatrix}
0&i\\
i&0
\end{bmatrix},\quad
\Lbf_2 =
\begin{bmatrix}
0&-1\\
1&0
\end{bmatrix},\quad
\Lbf_3 =
\begin{bmatrix}
i&0\\
0&-i
\end{bmatrix}.
\end{equation*}
\end{definition}
\vspace{4pt}


Here, we note that for any $\Xbf\in\SUT$, both $\Xbf$ and $-\Xbf$ map to the same element on $\SOT$. To see this, parametrize $\SUT$ by a unit vector $\qvec = (q_1,q_2,q_3,q_4)\Tr$, as
\begin{equation}
\Xbf = \begin{pmatrix}
q_1+iq_4   & -q_3 + iq_2\\
q_3 + iq_2 &q_1-iq_4
\end{pmatrix}\in \SUT,\label{eq:quaternion2specialunitary}
\end{equation}
which encompasses all of $\SUT$ by Definition~\ref{def:SUT}. Furthermore, we embed elements $\Xbf\in\SUT$ into $\R\in\SOT$ by 

\inACC{\begin{equation}\label{eq:quaternion2rotationmatrix}
\hspace{-2pt}\R=\hspace{-2pt}\begin{bmatrix}
q_1^2\con{+}q_2^2\con{-}q_3^2\con{-}q_4^2 \hspace{-0.5pt}&\hspace{-0.5pt}    2(q_2q_3\con{-}q_1q_4) \hspace{-0.5pt}&\hspace{-0.5pt}    2(q_2q_4\con{+}q_1q_3)\\
    2(q_2q_3\con{+}q_1q_4) \hspace{-0.5pt}&\hspace{-0.5pt} q_1^2\con{-}q_2^2\con{+}q_3^2\con{-}q_4^2 \hspace{-0.5pt}&\hspace{-0.5pt}      2(q_3q_4\con{-}q_1q_2)\\
    2(q_2q_4\con{-}q_1q_3) \hspace{-0.5pt} \hspace{-0.5pt}&\hspace{-0.5pt}     2(q_3q_4\con{+}q_1q_2) \hspace{-0.5pt}&\hspace{-0.5pt} q_1^2\con{-}q_2^2\con{-}q_3^2\con{+}q_4^2
\end{bmatrix}\hspace{-3pt}.
\end{equation}}
\inArxiv{\begin{equation}\label{eq:quaternion2rotationmatrix}
\hspace{-2pt}\R=\hspace{-2pt}\begin{bmatrix}
q_1^2{+}q_2^2{-}q_3^2{-}q_4^2 &    2(q_2q_3{-}q_1q_4) \hspace{-0.5pt}&\hspace{-0.5pt}    2(q_2q_4{+}q_1q_3)\\
    2(q_2q_3{+}q_1q_4) & q_1^2{-}q_2^2{+}q_3^2{-}q_4^2 &      2(q_3q_4{-}q_1q_2)\\
    2(q_2q_4{-}q_1q_3) &    2(q_3q_4{+}q_1q_2) & q_1^2{-}q_2^2{-}q_3^2{+}q_4^2
\end{bmatrix}\in\SOT.
\end{equation}}

\begin{definition}[Lie Maps]\label{def:groupactions}
Let $G$ be any of the above defined Lie groups, with algebra $\mathfrak{g}$. The $\emph{hat}$ map is denoted $[\cdot]_{G}^\land:\mathbb{R}^3\mapsto \mathfrak{g}$, the $\emph{vee}$ map is denoted $[\cdot]_{G}^\lor:\mathfrak{g}\mapsto \mathbb{R}^3$, and the associated exponential and logarithmic maps are denoted by $\text{Exp}_G:\mathfrak{g}\mapsto G$ and $\text{Log}_G:G\mapsto \mathfrak{g}$ respectively~\cite{hall2015lie,greiff2021similarities}.
\end{definition}

\inArxiv{
To simplify any implementation of the controllers, the 
\begin{definition}[Hat and vee maps of {\SOT}]
From Definitions~\ref{def:SOT} and~\ref{def:groupactions}, it follows that if $\Kbf = [\omegabf]_{\SOT}^{\land}\in\mathfrak{so}(3)$, then\label{def:veehatSOT}
\begin{equation*}
[\omegabf]_{\SOT}^{\land}  = \Sbf(\omegabf),\qquad 
[\Kbf]_{\SOT}^{\lor} = \begin{bmatrix}
[\Kbf]_{3,2}\\
[\Kbf]_{1,3}\\
[\Kbf]_{2,1}
\end{bmatrix}.
\end{equation*}
\end{definition}

\begin{definition}[Hat and vee maps of {\SUT}]
From Definitions~\ref{def:SOT} and~\ref{def:groupactions}, it follows that if $\Kbf = [\omegabf]_{\SUT}^{\land}\in\mathfrak{su}(2)$, then\label{def:veehatSUT}
\begin{equation*}
[\omegabf]_{\SUT}^{\land} = \begin{bmatrix}
i\omega_3 &  - \omega_2 + i\omega_1\\
\omega_2 + i\omega_1 & -i\omega_3
\end{bmatrix},\quad 
[\Kbf]_{\SUT}^{\lor} = \frac{1}{2}\begin{bmatrix}
      \Im([\Kbf]_{1,2} + [\Kbf]_{2,1})\\
      \Re([\Kbf]_{2,1} - [\Kbf]_{1,2})\\
      \Im([\Kbf]_{1,1} - [\Kbf]_{2,2})
\end{bmatrix}.
\end{equation*}
\end{definition}}

\begin{remark}
Here, we note that if $\Xbf\in\SUT$ and $\R\in\SOT$ represent the same attitude on \SOT, using the embedding in~\eqref{eq:quaternion2specialunitary} and~\eqref{eq:quaternion2rotationmatrix}, then, for any $\omegabf=\theta\uvec\in\mathbb{R}^3$,
\begin{align*}
&\Exp_{\SOT}([\omegabf]^\land_{\SOT})\in\SOT,&
&\Exp_{\SUT}([\omegabf/2]^\land_{\SUT})\in\SUT,
\end{align*}
both \change{represent the same attitude on \SOT, corresponding to a rotation of $\theta$ about a unit vector $\uvec\in\mathbb{R}^3$}. For any $\avec, \bvec\in\mathbb{R}^3$ with $\|\avec\| = \|\bvec\| $, we have the rotational composition
\begin{align}
\avec &= \R\bvec,&
\avec &= [\Xbf[\bvec]_{\SUT}^\land\Xbf\Hr]_{\SUT}^\lor.\label{eq:rotcomp}
\end{align}
\vspace{-12pt}

\end{remark}
\begin{definition}
In the following, we define a global frame $\{\Gcal\}$ spanned by three unit vectors $\evec_i$ with the $i$th element set to 1, and a body-fixed frame $\{\Bcal\}$ spanned by three unit vectors $\bvec_i$, which are related to the global frame by
\begin{equation}
\I=\begin{bmatrix}
\evec_1&\evec_2&\evec_3
\end{bmatrix}=\R\Tr
\begin{bmatrix}
\bvec_1&\bvec_2&\bvec_3
\end{bmatrix},
\end{equation}
where $\R$ rotates a vector from   $\{\Bcal\}$ to $\{\Gcal\}$\inArxiv{ (see Fig.~\ref{fig:geometry})}.
\end{definition}
In the following, we will at times refer to a \emph{reference rotation}, with $\R_r\in\SOT$ or $\Xbf_r\in\SUT$, with an associated set of body basis vectors $\bvec_{ri}$, and a \emph{desired rotation} $\R_d\in\SOT$ or $\Xbf_d \in\SUT$ with an associated set of body basis vectors $\bvec_{di}$ (for $i=1,2,3$). We also make frequent use of two closely related distances, here defined as follows.
\begin{definition}
Let $\Psi : \SOT\times\SOT \mapsto [0,2]$, where 
\begin{equation}\label{eq:distSOT}
    \Psi(\R_1, \R_2) = \frac{1}{2}\Trace(\I - \R_1\Tr\R_2).
\end{equation}\label{def:distSOT}
\end{definition}\vspace*{-10pt}
\begin{definition}\label{def:distSUT}
Let $\Gamma : \SUT\times\SUT \mapsto [0,2]$, where 
\begin{equation}\label{eq:distSUT}
    \Gamma(\Xbf_1, \Xbf_2) = \frac{1}{2}\Trace(\I - \Xbf_1\Hr\Xbf_2).
\end{equation}
\vspace{-12pt}

\end{definition}

\begin{subequations}\label{eq:FSFFull:contsys}
These distances facilitate elegant and powerful controller developments. The UAV dynamics are taken to be configured on $\SUT\times \Real^3$, with a state $(\pvec,\vvec,\Xbf,\omegabf)\in\Real^3\times\Real^3\times\SUT\times\Real^3$, and are to be driven along a reference $(\pvec_r,\vvec_r,\Xbf_r,\omegabf_r)\in\Real^3\times\Real^3\times\SUT\times\Real^3$, evolving \inArxiv{in time} by
\begin{align}
\dot\pvec_r &= \vvec_r &
\dot\pvec   &= \vvec  \\
m\dot\vvec_r &\hspace{4pt}=f_r\R_r\evec_3-mg\evec_3 & 
m\dot\vvec   &\hspace{4pt}=f\R\evec_3-mg\evec_3     \\
\dot\Xbf_r & = \Xbf_r[\omegabf_r/2]^{\land}_{\SUT} & 
\dot\Xbf   & = \Xbf[\omegabf/2]^{\land}_{\SUT} \label{eq:attitude}\\
\J\dot\omegabf_r&= \Sbf(\J\omegabf_r)\omegabf_r+\taubf_r & 
\J\dot\omegabf&= \Sbf(\J\omegabf)\omegabf+\taubf\label{eq:attituderate}
\end{align}
\end{subequations}
where $\R$ and $\R_r$ are computed by~\eqref{eq:rotcomp} from $\Xbf$ and $\Xbf_r$, respectively. In this notation, $\pvec\in\Real^3$ [m] defines the position of the UAV in $\{\Gcal\}$; $\vvec\in\Real^3$ [m/s] defines the velocity of the UAV in $\{\Bcal\}$; $\Xbf\in\SUT$ defines the attitude of the UAV; $\omegabf\in\Real^3$ [rad/s] denotes the usual attitude rates defined in $\{\Bcal\}$; $f>0$ [N] defines the thrust generated by the rotors; $\taubf$ [N$\cdot$m] denotes the torques defined in $\{\Bcal\}$ \inArxiv{(see Fig.~\ref{fig:geometry})}. The model is parameterized by a positive definite symmetric inertia matrix $\J\in\Real^{3\times 3}$, a positive mass $m>0$ [kg], and a constant positive gravitation acceleration $g>0$ [m/s$^2$].

\inArxiv{\begin{figure}
    \centering
    \includegraphics[width=0.4\textwidth]{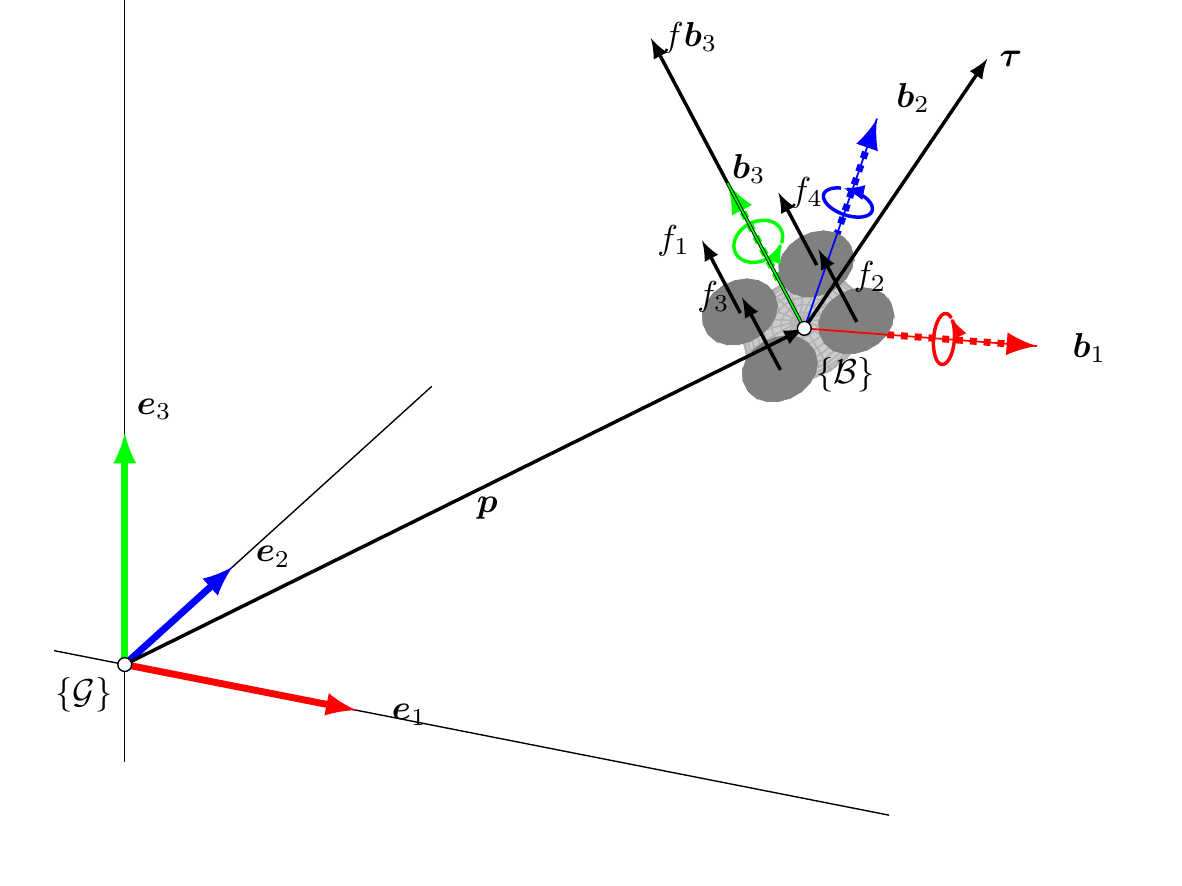}%
    \includegraphics[width=0.6\textwidth]{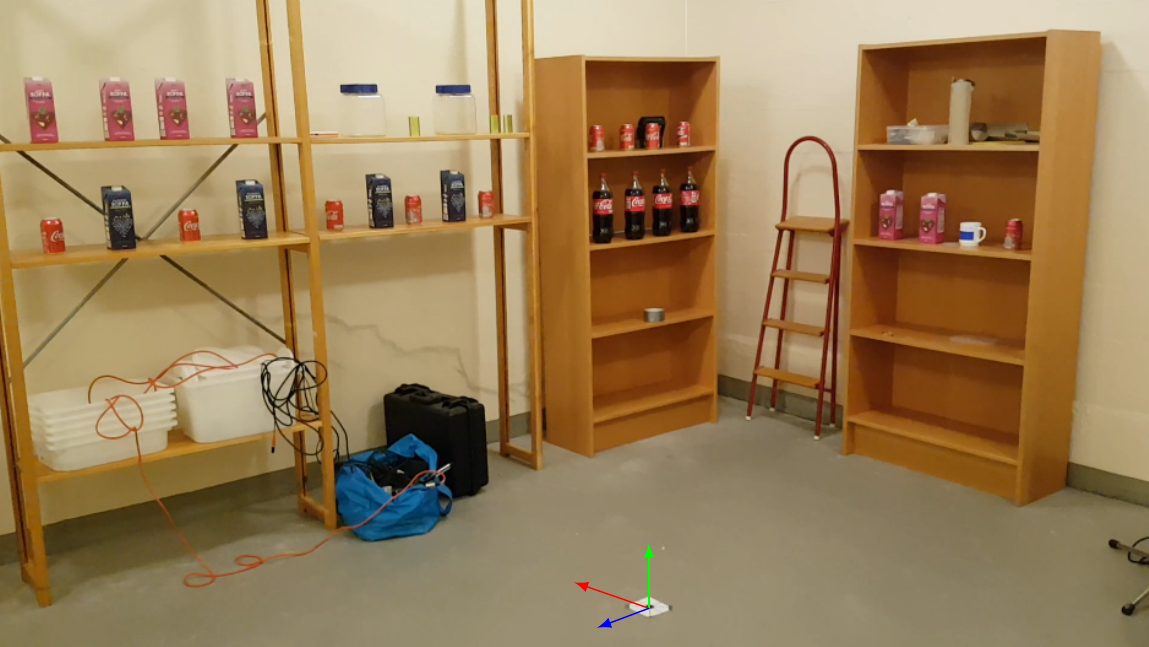}
    \vspace{-10pt}
    \caption{\textit{Left:} Depiction of the considered UAV geometry. \textit{Right:} The laboratory environment with $\{\Gcal\}$.}
    \label{fig:geometry}
\end{figure}}

For the purposes of this paper, we recall the attitude controller in \protect{\cite[Proposition 1]{greiff2021similarities}}, developed for the attitude subsystem characterized by~\eqref{eq:attitude} and~\eqref{eq:attituderate}, used to control $(\Xbf,\omegabf)$ along the reference trajectory $(\Xbf_r,\omegabf_r)$.
 
\begin{proposition}\label{thm:attitudeSUTcont}
\begin{subequations}
Let $\Xbf_e = \Xbf_r\Hr\Xbf\in\SUT$, and define
\inACC{\begin{align}
\evec_{\Xbf} = & \frac{1}{2}[\Xbf_e - \Trace(\Xbf_e)\I/2]^\lor_{\SUT}&&\in\mathbb{R}^3,\\
\evec_{\omegabf} =& \omegabf - [(\Xbf_e)\Hr[\omegabf_r]_{\SUT}^\land (\Xbf_e)]_{\SUT}^\lor &&\in\mathbb{R}^3,\\
\zvec =& (\|\evec_{\Xbf}\|,\|\evec_{\omegabf}\|)^\top&&\in\Real_{\geq 0}^2.
\end{align}}
\end{subequations}
\inArxiv{\begin{align}
&\evec_{\Xbf} =  \frac{1}{2}[\Xbf_e - \Trace(\Xbf_e)\I/2]^\lor_{\SUT}\in\mathbb{R}^3,&&
&\evec_{\omegabf} = \omegabf - [(\Xbf_e)\Hr[\omegabf_r]_{\SUT}^\land (\Xbf_e)]_{\SUT}^\lor \in\mathbb{R}^3,
\end{align}}
\inACC{\begin{subequations}\label{eq:SU2:contmatrices}
Take a set of gains $k_X, k_\omega,k_c>0$ such that the matrices
\begin{align}
\hspace{-2pt}\W^{aa}&=\begin{bmatrix}
\frac{k_ck_X}{\maxeig(\J)}&
-\frac{k_ck_w}{2\mineig(\J)}\\
-\frac{k_ck_w}{2\mineig(\J)}&
k_\omega - \frac{k_c}{4}\\
\end{bmatrix}&&\succ \Z,\\
\hspace{-4pt}\M_1^{aa}&= \frac{1}{2}\begin{bmatrix}
4k_X & -k_c\\ -k_c & \mineig(\J)
\end{bmatrix}&&\succ \Z,\\
\M_2^{aa}&=\frac{1}{2}\begin{bmatrix}
\frac{8k_X}{2-\phi} & k_c\\ k_c & \maxeig(\J)
\end{bmatrix}&&\succ \Z.
\end{align}
\end{subequations}}
\inArxiv{\noindent and let $\zvec = (\|\evec_{\Xbf}\|,\|\evec_{\omegabf}\|)^\top\in\Real_{\geq 0}^2$. Take a set of gains $k_X, k_\omega,k_c>0$ such that the matrices
\begin{align}\label{eq:SU2:contmatrices}
&\hspace{-4pt}\W^{aa}=\begin{bmatrix}
\frac{k_ck_X}{\maxeig(\J)}&
-\frac{k_ck_w}{2\mineig(\J)}\\
-\frac{k_ck_w}{2\mineig(\J)}&
k_\omega - \frac{k_c}{4}\\
\end{bmatrix}\succ \Z,&&
&\hspace{-4pt}\M_1^{aa}= \frac{1}{2}\begin{bmatrix}
4k_X & -k_c\\ -k_c & \mineig(\J)
\end{bmatrix}\succ \Z,&&
&\hspace{-4pt}\M_2^{aa}=\frac{1}{2}\begin{bmatrix}
\frac{8k_X}{2-\phi} & k_c\\ k_c & \maxeig(\J)
\end{bmatrix}\succ \Z.
\end{align}}

Then, for any initial error on the domain
\begin{align}\label{eq:thm:attitudeSUTcont:domain}
D^a {=}\Bigg\{
\begin{bmatrix}
\evec_{\Xbf}\\ \evec_{\omegabf}
\end{bmatrix}\in\Real^6\Bigg|
\begin{matrix}
\Gamma(\Xbf_r(\tz), \Xbf(\tz))\leq \phi<2,\\
\zvec(\tz)\Tr
\M_2^{aa}
\zvec(\tz)\leq k_X\phi
\end{matrix}
\Bigg\},
\end{align}
driving the system~\eqref{eq:attituderate} with a full-state feedback
\inACC{\begin{subequations}\label{eq:thm:SU2:feedback}
\begin{align}
\taubf = {-} k_X& \evec_{\Xbf} {-} k_\omega \evec_{\omegabf} {-}\Sbf(\J\omegabf)\omegabf\\ +\J[&-[\evec_{\omegabf}/2]_{\SUT}^\land\Xbf_e\Hr[\omegabf_r]_{\SUT}^\land \Xbf_e\label{eq:ffa}\\
&+ \Xbf_e\Hr[\dot{\omegabf}_r]_{\SUT}^\land \Xbf_e\\
&+ \Xbf_e\Hr[\omegabf_r]_{\SUT}^\land \Xbf_e[\evec_{\omegabf}/2]_{\SUT}^\land]_{\SUT}^\lor\label{eq:ffc}
\end{align}
\end{subequations}}
\inArxiv{\begin{subequations}\label{eq:thm:SU2:feedback}
\begin{align}
\taubf = {-} k_X& \evec_{\Xbf} {-} k_\omega \evec_{\omegabf} {-}\Sbf(\J\omegabf)\omegabf +\J[-[\evec_{\omegabf}/2]_{\SUT}^\land\Xbf_e\Hr[\omegabf_r]_{\SUT}^\land \Xbf_e
+ \Xbf_e\Hr[\dot{\omegabf}_r]_{\SUT}^\land \Xbf_e
+ \Xbf_e\Hr[\omegabf_r]_{\SUT}^\land \Xbf_e[\evec_{\omegabf}/2]_{\SUT}^\land]_{\SUT}^\lor
\end{align}
\end{subequations}}
yields an equilibrium point $\zvec=\Z \Rightarrow (\evec_{\Xbf}, \evec_{\omegabf}) = (\Z, \Z)\Rightarrow (\Xbf, \omegabf)=(\Xbf_r, \omegabf_r)$, which is UES on $D^a$.
\end{proposition}
\inArxiv{
\begin{proof}
The proof is given in~\protect{\cite[Proposition 1]{greiff2021similarities}}, and follows by the definition of a Lyapunov function candidate
\begin{equation}\label{eq:Lyapufunccand}
    \Lyap^a= k_X\Gamma(\Xbf_r, \Xbf) + c_a\evec_{\omegabf}\cdot \evec_{\Xbf} + \frac{1}{2}\evec_{\omegabf}\cdot \J\evec_{\omegabf},
\end{equation}
where all solutions remain on $D^a$ as $\Lyap^a|_{c_a=0}$ is non-increasing on $D^a$, also ensuring that $\Gamma(\Xbf_r,\Xbf)<\phi$ for all $t\geq \tz$. As $\|\zvec\|_{\M_1^{aa}}^2 \leq \Lyap^a\leq \|\zvec\|_{\M_2^{aa}}^2$ and $\dot{\Lyap}^a\leq - \|\zvec\|_{\W^{aa}}^2$, a proof of UES on $D^a$ follows by  ~\protect{\cite[Theorem 4.10]{khalil2002nonlinear}}.
\end{proof}
}
With these preliminaries, we proceed by posing and solving the control problem for the full UAV dynamics.

\section{THE UAV CONTROL PROBLEM}\label{sec:control}
We start by defining the control problem in Section~\ref{sec:controlproblem}. To solve this problem, intuition regarding the problem geometry and the controller design is given in Section~\ref{sec:geomintuition} and~\ref{sec:contIntuition}, respectively. Based on these ideas, and using the geometric controller in~\cite{greiff2021similarities} as a starting point, an analogous continuous feedback to that in~\cite{lee2010geometric} is derived for the system in~\eqref{eq:FSFFull:contsys} configured on $\SUT\times \Real^3$. This main theoretical result is given in Section~\ref{sec:contSUT}.

\subsection{Control Problem}\label{sec:controlproblem}

\begin{problem}\label{problem:FSFFull}
Consider a system with a state $\xvec=(\pvec, \vvec,\Xbf,\omegabf)\in \Real^3\times\Real^3\times \SUT\times \Real^3$, with an associated reference trajectory $\xvec_r = (\pvec_r, \vvec_r,\Xbf_r,\omegabf_r)\in\Real^3\times\Real^3\times \SUT\times \Real^3$, driven by $(f,\taubf)\in\mathbb{R}_{\geq 0}\times \Real^3$ and $(f_r,\taubf_r)\in\mathbb{R}_{\geq 0}\times \Real^3$, respectively, evolving by the UAV dynamics in~\eqref{eq:FSFFull:contsys}. Assume that the state $\xvec$ is known and design a full state feedback law such that $\xvec\to\xvec_r$ as $t\to\infty$, and characterize the stability properties of the closed-loop system.
\end{problem}

\subsection{Geometric Intuition}\label{sec:geomintuition}
To provide some intuition regarding the distance in~\eqref{eq:distSUT} and aid the developments, consider two elements $\Xbf_1, \Xbf_2\in\SUT$. Take the conjugate product $\Xbf_e=\Xbf_1\Hr\Xbf_2$ to be a rotation of $\theta$ about a unit vector $\uvec=(u_1,u_2,u_3)\Tr$, as
\begin{equation*}
\Xbf_e \triangleq
\begin{pmatrix}
\cos(\tfrac\theta2)+iu_3\sin(\tfrac\theta2)   & -u_2\sin(\tfrac\theta2) + iu_1\sin(\tfrac\theta2)\\
u_2\sin(\tfrac\theta2) + iu_1\sin(\tfrac\theta2) &\cos(\tfrac\theta2)-iu_3\sin(\tfrac\theta2)
\end{pmatrix},
\end{equation*}
then $\Gamma(\Xbf_1, \Xbf_2)= 1 - \cos(\tfrac\theta2)$, and $\forall\Gamma(\Xbf_1, \Xbf_2)\leq\phi<2$,
\begin{equation}\label{eq:boundsdistSU2}
\frac{1}{2}\sin^2(\tfrac\theta2) \leq \Gamma(\Xbf_1, \Xbf_2)\leq \frac{1}{2-\phi}\sin^2(\tfrac\theta2).
\end{equation}
for any $\phi$ upper bounding $\Gamma(\Xbf_1, \Xbf_2)$ \inArxiv{(see Figure~\ref{fig:geomintA})}. This will be used in the stability analysis of the controllers on $\SUT$. Specifically, \change{if $\evec_{\Xbf} =  \frac{1}{2}[\Xbf_e - \Trace(\Xbf_e)\I/2]^\lor_{\SUT}=\frac{1}{2}\sin(\tfrac\theta2)\uvec$ with $\Xbf_e = \Xbf_1^*\Xbf_2$, $\|\evec_{\Xbf}\|^2 = \frac{1}{4}\sin^2(\tfrac\theta2)$, then~\eqref{eq:boundsdistSU2} becomes}
\begin{equation}\label{eq:boundsdistSU2B}
2\|\evec_{\Xbf}\|^2 \leq \Gamma(\Xbf_1, \Xbf_2)\leq \frac{4}{2-\phi}\|\evec_{\Xbf}\|^2.
\end{equation}
\inArxiv{
\begin{figure}
    \centering
    \inACC{\includegraphics[width=\textwidth]{figures/distbound_updated.png}}
    \inArxiv{\includegraphics[width=\textwidth]{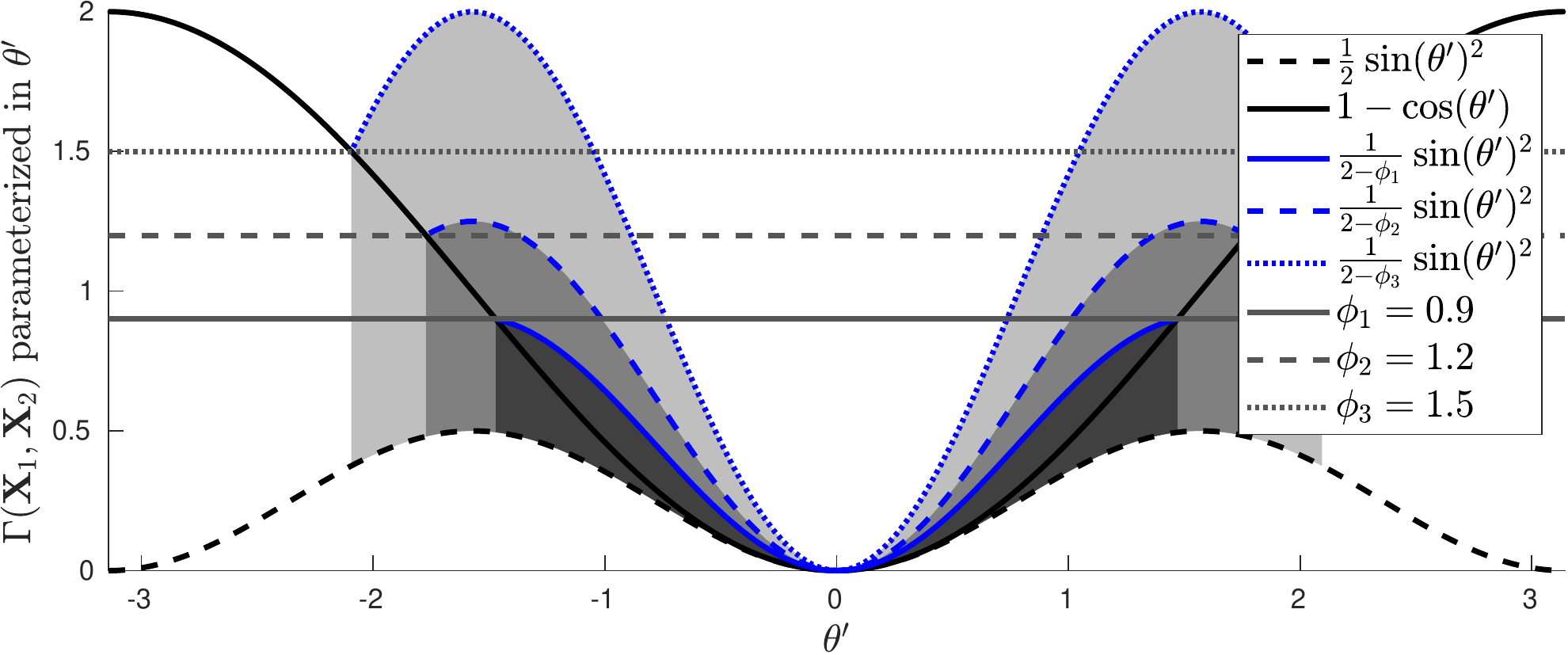}}
\caption{Illustration of the distance $\Gamma$ with the element $\Xbf_1, \Xbf_2$ representing a rotation of $\theta$ about an arbitrary unit axis, with $\theta^{\prime}=\theta/2$. Illustrates the upper and lower bounds in~\eqref{eq:boundsdistSU2} for three different values of $\phi \in\{0.9,1.2,1.5\}$.}
    \label{fig:geomintA}
\end{figure}}

This observation leads to several insightful geometric relationships with respect to the sine and cosine of the angle of the eigen-axis rotation between two vectors. For future reference, we first make some preliminary geometric observations with respect to the vectors $\bvec_{d3}=\R_d\evec_3$ and $\bvec_{3}=\R\evec_3$, and their cosine angle $\bvec_{d3}\cdot \bvec_{3} = \cos(\theta)$.
\begin{itemize}
\item Firstly, we note that
\begin{align}\label{eq:firstobs}
\Gamma(\Xbf_d,\Xbf) &\leq \phi < 1-1/\sqrt{2} \Rightarrow \cos(\theta) > 0.
\end{align}
\item Secondly, as $\sin^2(\theta)\leq 4\sin^2(\theta/2)$ for all $\cos(\theta)>0$,
\begin{align}\label{eq:secondobs}
\hspace{-5pt}\sin^2(\theta)\leq
2^4\|\evec_{\Xbf}\|^2\leq 2^3\Gamma(\Xbf_d,\Xbf)\leq 2^3\phi\triangleq \alpha^2.
\end{align}
\end{itemize}
These geometric relationships are illustrated in~\Fig~\ref{fig:geometrySU2}, and from these observations, it is clear that for any
\begin{equation}
\Gamma(\Xbf_d,\Xbf)\leq \phi<2^{-3}
\Rightarrow \cos(\theta)>0\Rightarrow 
\alpha<1
\end{equation}
\begin{figure}
    \centering
    \inACC{\includegraphics[width=\columnwidth]{figures/example_dustances_SU2_FSFFullfig.png}
\vspace{-20pt}}
\inArxiv{\includegraphics[width=\columnwidth]{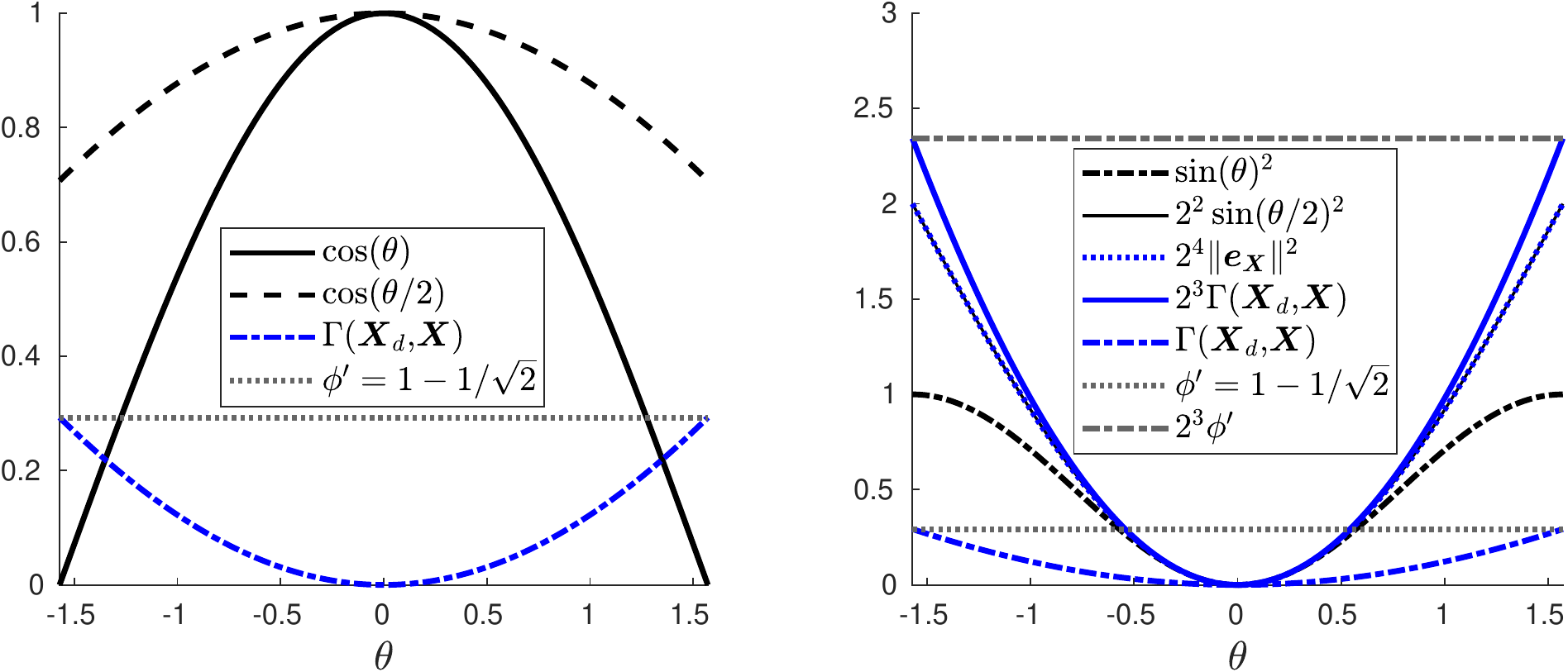}}
\caption{Illustration of the geometric relationships that facilitate the stability proof. \textit{Left:} Expressions relating to the cosine of the rotation angle $\theta$. \textit{Right:} Expressions relating to the squared sine of the rotation angle $\theta$.}
    \label{fig:geometrySU2}
\end{figure}

\inArxiv{\begin{remark}
The condition on $\phi$ may seem overly restrictive, but it is worth noting that if $\Xbf_e \in\SUT$ is parameterized in the more common ZYX Tait-Bryan rotation angles (pitch, roll, and yaw), as $\boldsymbol{\eta}\in\Real^3$, then $\|\etabf\|^2\approx 4\Psi(\R_d,\R) \approx 8\Gamma(\Xbf_d, \Xbf)$. As such, this condition in \eqref{eq:secondobs} is more restrictive than $\Psi(\R_d,\R)\leq 1$ in~\cite{lee2010control} (by approximately a factor of four). But \eqref{eq:secondobs} still permits significant attitude errors when considered in the Tait-Bryan rotation angles.
\end{remark}}

\subsection{Controller Intuition}\label{sec:contIntuition}
A solution to Problem~\ref{problem:FSFFull} when considered on $\SOT\times \Real^3$ is given in~\cite{lee2010geometric}, with more detail in~\cite{lee2010control}. The developments here follow a similar idea, according to  the observations in Section~\ref{sec:geomintuition}. Consider a set of translation control errors, defined as $\evec_{\pvec}=\pvec - \pvec_r\in\Real^3$ and $\evec_{\vvec} = \vvec - \vvec_r\in\Real^3$. Combined with the reference acceleration, these errors can be used to compute a desired force in the frame $\{\Gcal\}$, as
\begin{equation}\label{eq:FSF:desforce}
\fvec_d = -k_p\evec_{\pvec}-k_v\evec_{\vvec}+mg\evec_3 + m\ddot\pvec_r.
\end{equation}

As the controlled system in~\eqref{eq:FSFFull:contsys} is only capable of generating forces along the $\bvec_3$-direction, its attitude needs to be controlled to a desired attitude, $\R_d(t)\in\SOT$, which transiently may differ from $\R_r(t)\in\SOT$ when correcting for the errors in $\evec_{\pvec}$ and $\evec_{\vvec}$. Consider expressing this attitude in terms of a set of desired body basis vectors $\bvec_{d,i}\in\Real^3$, forming a desired body frame $\{\Bcal_d\}$. It is clear that $\bvec_{d,3}=\fvec_d/\|\fvec_d\|$, but the final degree of freedom can be fixed in many ways. Three such examples are given below, where:
\begin{itemize}
\item[(i)] the vector $\bvec_{d,1}$ is provided explicitly;
\item[(ii)] the vector $\bvec_{d,1}$ is computed from $\R_r$;
\item[(iii)] the vector $\bvec_{d,1}$ is defined with respect to $\bvec_{d,3}$.
\end{itemize}

In the case (ii), the desired body direction can be computed through a sequence of projections as outlined in~\cite{lee2010geometric}, where
\begin{equation}\label{eq:bvec1SO3}
\bvec_{d,1} = -\frac{1}{\|\bvec_{d,3}\times \bvec_{r,1}\|}(\bvec_{d,3}\times(\bvec_{d,3}\times \bvec_{r,1})).
\end{equation}
This permits a construction of the desired rotation, attitude rate and attitude rate time-derivative for case (i) and (ii), as
\begin{subequations}\label{eq:FSF:computedesref}
\begin{align}
\R_d &= \begin{bmatrix}
\bvec_{d,1} & (\bvec_{d,3}\times \bvec_{d,1}) & \bvec_{d,3}
\end{bmatrix}\in\SOT,\\
\omegabf_d &= [\R_d\Tr\dot\R_d]_{\SOT}^{\lor}\in\Real^3,\label{eq:SO3dera}\\
\dot\omegabf_d &=  [\dot\R_d\Tr\dot\R_d+\R_d\Tr\ddot\R_d]_{\SOT}^{\lor}\in\Real^3.\label{eq:SO3derb}
\end{align}
\end{subequations}
In the case (iii), the desired attitude can be formed by computing an angle $\beta \triangleq \mathrm{atan2}(\sqrt{f_{d1}^2 + f_{d2}^2}, f_{d3})$, and defining
\begin{subequations}
\begin{align}\label{eq:defqd}
\nvec  &= (f_{d1}^2 + f_{d2}^2)^{-1/2}(-f_{d2},f_{d1},0)\Tr\\
\Xbf_A &= \Exp_{\SUT}([\beta\nvec/2]^\land_{\SUT})\\
\Xbf_B &= \Exp_{\SUT}([\psi_r\evec_3/2]^\land_{\SUT})
\end{align}
\end{subequations}
where $\psi_r$ parameterizes a desired rotation about the body $\bvec_{d,3}$ vector. With these definitions, the desired attitude can be defined as the composition $\Xbf_d = \Xbf_A\Xbf_B$, computing the desired rates and accelerations through the inverse kinematics, similar to~\eqref{eq:SO3dera} and~\eqref{eq:SO3derb}, or by numerical differentiation. If considering the continuous attitude controller in Proposition~\ref{thm:attitudeSUTcont}, this desired attitude needs to be chosen with care such that $\Xbf_d$ is continuous in time to avoid dynamical unwinding~\cite{mayhew2011quaternion}. This can be done using the ideas in~\cite{mayhew2012path}, or directly using the distance in Definition~\ref{def:distSUT}.

{\begin{remark}\label{rem:compXd}
In a controller implementation running at a time-step of $h$ [s], enforcing continuity in $\Xbf_d(t)$ can be done at a time $t$ by computing one of the two elements $\bar\Xbf_d(t)\in\SUT$ associated with $\R_d(t)\in\SOT$, taking
\begin{equation*}
\Xbf_d(t)=
\begin{cases}
+\bar\Xbf_d(t),&\text{if}\;\;\Gamma(\bar\Xbf_d(t),\Xbf_d(t-h))<1\\
-\bar\Xbf_d(t),&\text{otherwise}
\end{cases}.
\end{equation*}
The computation of $\bar\Xbf_d$ from $\R_d$ can be done though~\eqref{eq:rotcomp}.
\end{remark}}

\subsection{Continuous feedback on $\SUT\times \Real^3$}\label{sec:contSUT}
With this intuition, the main result is stated as follows.
\begin{proposition}[Geometric Control on $\SUT\times\Real^3$]\label{prop:geomSUT}
Consider the dynamics \eqref{eq:FSFFull:contsys} controlled by a feedback where:
\begin{itemize}
\item the torques, $\taubf$, are computed by the controller Proposition~\ref{thm:attitudeSUTcont} implemented to track a trajectory $({\Xbf}_d,{\omegabf}_d,\dot\omegabf_d)$;
\item the desired attitude reference trajectory is formed by \eqref{eq:FSF:computedesref}, expanding $\R_d$ or $\qvec_d$ into $\Xbf_d\in \SUT$, and enforcing continuity of the desired reference on $\SUT$;
\item the actuating force is computed as $f = \fvec_d\cdot \R\evec_3$, with the desired force computed as described in~\eqref{eq:FSF:desforce}.
\end{itemize}
Assume that, for all $t\geq \tz$:
\begin{itemize}
\item[(A1)] the desired reference $({\Xbf}_d,{\omegabf}_d,\dot\omegabf_d)$ is well defined;
\item[(A2)] there exists a bound $\|mg\evec_3+m\ddot\pvec_r\|\leq B_f$;
\item[(A3)] and $\|\evec_{\pvec}(\tz)\|<B_p$ for a fixed $B_p>0$;
\item[(A4)] the initial errors satisfy $\Gamma(\Xbf_d(\tz),\Xbf(\tz))\leq \phi< 2^{-3}$;
\item[(A5)] the controller parameters $(k_p,k_v,k_X,k_\omega,c_a,c_p)\hspace{-1pt}\in\hspace{-1pt}\Real^{6}_{>0}$ are chosen such that for $\alpha = 2\sqrt{2\phi}$, the matrices
\inACC{\begin{align}\label{eq:FSFFullMatrices}
\M_1^{pp}&\triangleq 
\frac{1}{2}\begin{bmatrix}
k_p&-c_p\\
\star & m
\end{bmatrix},\quad
\M_2^{pp}\triangleq 
\frac{1}{2}\begin{bmatrix}
k_p&c_p\\
\star & m
\end{bmatrix},
\\
\W^{pp}&\triangleq 
\begin{bmatrix}
\frac{c_pk_p}{m}(1 - \alpha)&-\frac{c_pk_v}{2m}(1 + \alpha)\\
\star & k_v(1-\alpha)-c_p
\end{bmatrix},\nonumber
\end{align}}
\inArxiv{\begin{align}\label{eq:FSFFullMatrices}
\M_1^{pp}&\triangleq 
\frac{1}{2}\begin{bmatrix}
k_p&-c_p\\
\star & m
\end{bmatrix},\quad
\M_2^{pp}\triangleq 
\frac{1}{2}\begin{bmatrix}
k_p&c_p\\
\star & m
\end{bmatrix},&
\W^{pp}&\triangleq 
\begin{bmatrix}
\frac{c_pk_p}{m}(1 - \alpha)&-\frac{c_pk_v}{2m}(1 + \alpha)\\
\star & k_v(1-\alpha)-c_p
\end{bmatrix},
\end{align}}
are all positive definite, and there exist a matrix
\begin{align}\label{eq:FSFFullMatrices2}
\W^{pa} &\triangleq 4\begin{bmatrix}
\frac{B_fc_p}{m} & 0 \\
B_f + k_pB_p & 0
\end{bmatrix},
\end{align}
such that $B_z=4\mineig(\W^{aa})\mineig(\W^{pp})-\|\W^{pa}\|^2>0$.
\end{itemize}
Consider a domain
\begin{align}\label{eq:thm:FSFSUTcont:domain}
\inACC{\hspace{-2pt}}D \inACC{\hspace{-2pt}}=\inACC{\hspace{-2pt}}
\begin{Bmatrix}\inACC{\hspace{-3pt}}
\begin{bmatrix}
\evec_{\pvec}(\tz)\\
\evec_{\vvec}(\tz)\\
\evec_{\Xbf}(\tz)\\
\evec_{\omegabf}(\tz)
\end{bmatrix}
\inACC{\hspace{-2pt}}\in\inACC{\hspace{-2pt}}\Real^{12}\inArxiv{\hspace{3pt}}\vrule\;
\begin{matrix}
\Gamma(\Xbf_d(\tz), \Xbf(\tz))=\phi_{\circ}\leq \phi<2^{-3},\\
\|\evec_{\omegabf}(\tz)\|^2\leq \frac{2}{\maxeig(\J)}k_X(\phi-\phi_{\circ}),\\
\maxeig(\M_2^{aa})\|\zvec_a(\tz)\|^2+\hspace{1.5pt}\hphantom{\frac{1}{2}k_pB_p^2}\\
\maxeig(\M_2^{pp})\|\zvec_p(\tz)\|^2
\leq \frac{1}{2}k_pB_p^2
\end{matrix}\inACC{\hspace{-2pt}}
\end{Bmatrix}\inACC{\hspace{-2pt}},
\end{align}
where $\zvec_p=(\|\evec_{\pvec}\|;\|\evec_{\vvec}\|)\in\Real^{2}_{\geq 0}, \zvec_a=(\|\evec_{\Xbf}\|;\|\evec_{\omegabf}\|)\in\Real^{2}_{\geq 0}$. Given the assumptions (A1)-(A5), the equilibrium point $(\evec_{\pvec},\evec_{\vvec},\evec_{\Xbf},\evec_{\omegabf})=(\Z,\Z,\Z,\Z)$ is UES on $D$.
\end{proposition}
\begin{proof}
\inACC{The proof is given in the extended paper, with a sketch provided here.}\inArxiv{The proof is given in the Appendix, with a sketch provided here.} Similar to \protect{\cite[Proposition 2]{lee2010geometric}}, it follows by defining a Lyapunov function candidate
\inACC{\begin{subequations}
\begin{align}\label{eq:FSF:full:SU2:lyap}
\hspace{-1pt}\bar\Lyap =& \frac{1}{2}k_p\|\evec_{\pvec}\|^2 + \frac{1}{2}m\|\evec_{\vvec}\| ^2 + c_p \evec_{\pvec}\cdot \evec_{\vvec} +\\
&k_X\Gamma(\Xbf_d,\Xbf) + c_a\evec_{\Xbf}\cdot \evec_{\omegabf} + \frac{1}{2}\evec_{\omegabf}\cdot \J\evec_{\omegabf}.\hspace{-1pt}
\end{align}
\end{subequations}}
\inArxiv{\begin{subequations}
\begin{align}\label{eq:FSF:full:SU2:lyap}
\hspace{-1pt}\bar\Lyap =& \frac{1}{2}k_p\|\evec_{\pvec}\|^2 + \frac{1}{2}m\|\evec_{\vvec}\| ^2 + c_p \evec_{\pvec}\cdot \evec_{\vvec} +
k_X\Gamma(\Xbf_d,\Xbf) + c_a\evec_{\Xbf}\cdot \evec_{\omegabf} + \frac{1}{2}\evec_{\omegabf}\cdot \J\evec_{\omegabf}.\hspace{-1pt}
\end{align}
\end{subequations}}
Given the assumptions (A1)-(A5), it is shown all solutions initialized on $D$ remain on this domain for all $t\geq \tz$. Furthermore, it is shown $\bar\Lyap$ is continuously differentiable, and there exist constants $c_1,c_2,c_3>0$ expressed in the matrices in~\eqref{eq:FSFFullMatrices} and~\eqref{eq:FSFFullMatrices2}, such that
\begin{equation}
c_1\|\bar\zvec\|^2\leq \bar\Lyap\leq c_2\|\bar\zvec\|^2,\quad (\der/\der t){\bar\Lyap}\leq -c_3\|\bar\zvec\|^2,
\end{equation}
where $\bar\zvec=(\|\evec_{\pvec}\|;\|\evec_{\vvec}\|;\|\evec_{\Xbf}\|;\|\evec_{\omegabf}\|)$. This holds for all solutions of the error dynamics on $D$. Applying~\protect{\cite[Theorem 10]{khalil2002nonlinear}} shows UES of $\bar\zvec=\Z$ on $D$.
\end{proof}

\inArxiv{\subsection{Comments on the Assumptions}
\subsubsection{Assumption (A1)} The assumption is generally difficult to guarantee, as the denominators in~\eqref{eq:bvec1SO3} and~\eqref{eq:defqd} depend on the control errors, and are defined with respect to the reference trajectory and not the desired reference trajectory. There may exist solutions both initially and transiently where these are ill defined. However, as the full-state information of the UAV is available, such cases can easily be detected and handled in the controller implementation (see, e.g.,~\cite{greiff2018nonsingular}).

\subsubsection{Assumption (A4)}} It may seem as though Assumption (A4) is restrictive, as it only permits small attitude errors. It can be relaxed slightly;  given the characterization of the domain of exponential convergence in Proposition~\ref{thm:attitudeSUTcont}, the following holds.

\begin{proposition}
Consider the system in~\eqref{eq:FSFFull:contsys} in closed-loop feedback with Proposition~\ref{prop:geomSUT}, but instead of Assumption (A4), assume that the initial errors satisfy \label{prop:FSFFull:attractivitySU2}
\begin{subequations}\label{eq:domain2}
\begin{align}
&\Gamma(\Xbf_d(\tz), \Xbf(\tz))\leq \phi<2,\\
&\|\evec_{\omegabf}(\tz)\|^2\leq \frac{2}{\maxeig(\J)}k_X(\phi-\Gamma(\Xbf_d(\tz), \Xbf(\tz))).
\end{align}
\end{subequations}
Under these conditions, the origin $(\evec_{\pvec},\evec_{\vvec},\evec_{\Xbf},\evec_{\omegabf})=(\Z,\Z,\Z,\Z)$ is asymptotically attractive.
\end{proposition}

\begin{proof}
This becomes completely analogous to the proof in \protect{\cite[Appendix E]{lee2010geometric}}, therefore omitted for brevity. It follows by showing boundedness of solutions on $t\in[\tz,t^*]$, before the errors approach $D$ as defined in~\eqref{eq:thm:FSFSUTcont:domain} at a finite time $t^*$, after which the errors decay exponentially to the origin.
\end{proof}

It is worth noting that  the UAV system with the attitude controller in Proposition~\ref{thm:attitudeSUTcont} can be shown to be almost globally asymptotically stable (AGAS), in the sense that all initial conditions converge to a set $\Gamma(\Xbf_d(\tz), \Xbf(\tz))\in \{0,2\}$ with $\evec_{\omegabf}=\Z$, corresponding to a stable point interior of~\eqref{eq:domain2}, or an unstable point on the boundary of the domain of exponential attraction in~\eqref{eq:domain2}. As such, the solutions associated with \emph{almost all} initial conditions  asymptotically converge to~\eqref{eq:domain2}, with subsequent convergence of the errors to $D$ in~\eqref{eq:thm:FSFSUTcont:domain}.

\inArxiv{\subsubsection{Assumption (A5)}  As pointed out in~\cite{greiff2021similarities}, for small $c_a$, the matrices $\M_1^{aa}, \M_2^{aa}, \W^{aa}$ are positive definite. Similarly, for sufficiently small $c_p$, the matrices $\M_1^{pp}, \M_2^{pp}, \W^{pp}$ are positive definite. Specifically, sufficient conditions are
\begin{align*}
    c_a&<\min\begin{Bmatrix}
4k_{\omega}, 
\dfrac{4k_\omega k_X \mineig(\J)^2}{\maxeig(\J)k_w^2+\mineig(\J)^2k_X},
2\sqrt{k_X \mineig(\J)}
\end{Bmatrix}\\
    c_p&<\min\begin{Bmatrix}
k_{v}(1- \alpha), 
\dfrac{4 m k_p k_v(1 - \alpha)^2}{k_v^2(1 +\alpha)^2+4 m k_p(1-\alpha)},
\sqrt{k_x m}
\end{Bmatrix}
\end{align*}
Similarly, it is clear that the last condition in Assumption (A5) can be satisfied by decreasing $\phi$ and $\alpha$ and/or increasing the tuning parameters $k_X$ and $k_{\omega}$ in relation to the $B_f$ and $B_p$. It should also be noted that it can be replaced by a less restrictive condition $\W^{pp} - \W^{pa}(\W^{aa})^{-1}(\W^{pa})\Tr\succ \Z$.}

\section{THE ESTIMATION PROBLEM}\label{sec:estimation}

\inArxiv{In this section, we present a system for simultaneously localization and mapping (SLAM), which is used to generate real-time pose estimates of the UAV from a monocular video feed. The system is feature based, relying on extracted ORB features \cite{orb} for tracking and matching. Additionally it utilizes accelerometer data and gyroscopic data from an inertial measurement unit (IMU) to constrain the scale and two rotational components, the pitch and the roll. These are otherwise ambiguous if only images are used. This is important since the controller assumes positions in meters and rotations relative to the gravity direction. The system uses the method proposed in \cite{preintegration} to accumulate IMU data into so-called deltas, containing information about the metric relative transformation between consecutive frames. These deltas are then used together with feature matches between images in a large non-linear optimization problem called bundle adjustment~\cite{isam2}, to solve for the camera poses and three-dimensional structure up to a metric solution in $\{\Gcal\}$.


The SLAM system is divided into modules, each performing a specific task concurrently with the others and communicates with the other modules using message passing. The main modules in the system are described in order:
\begin{itemize}
    \item[1] The initialization module (see Section~\ref{sec:SLAM:initialization});
    \item[2] The tracking module (see Section~\ref{sec:SLAM:tracking});
    \item[3] The re-localization module (see Section~\ref{sec:SLAM:relocalization});
    \item[4] The triangulation module (see Section~\ref{sec:SLAM:triangulation});
    \item[5] The mapping module (see Section~\ref{sec:map}).
\end{itemize}

Furthermore, the integration of the SLAM system with the EKF on the Crazyflie is described in Section~\ref{sec:SLAM:integration}.}


\inArxiv{\begin{figure}[t!]
\resizebox{\textwidth}{!}{\begin{tikzpicture}[
module/.style={
  draw=#1,
  rounded corners=8pt,
  line width=1pt,
  align=center,
  text width=3cm,
  minimum height=1cm,
  font=\strut\sffamily
  },
myarr/.style={
  ->,
  >=latex,
  thick
  },
node distance=0.5cm   
]
\newcommand{\relocalizationtext}{\hrule\vspace{2pt}
    \footnotesize (i) Perceptual hashing~\cite{Zauner2010ImplementationAB}\\
    \footnotesize (ii) Brute-force matching\hspace{8pt}$\;$
}
\newcommand{\trackingtext}{\hrule\vspace{2pt}
    \footnotesize (i) ORB feature matching~\cite{orb}\\
    \footnotesize (ii) Form IMU-deltas~\cite{preintegration}\hspace{20pt}$\;$
    \footnotesize (iii) Pose estimation\hspace{40pt}$\;$
}
\newcommand{\triangulationtext}{\hrule\vspace{2pt}
    \footnotesize (i) Depth filters \cite{VOGIATZIS2011434}\hspace{22pt}$\;$
}
\newcommand{\mappingtext}{\hrule\vspace{2pt}
    \hspace*{-3pt}\footnotesize (i) Key-frame insertion/pruning \hspace{-1pt}$\;$\hspace*{-3pt}\\
    \hspace*{-3pt}\footnotesize (ii) Validation of 3D points \hspace{17pt}$\;$\hspace*{-3pt}\\
    \hspace*{-3pt}\footnotesize (iii) Local bundle adjustment \cite{isam2}\hspace*{-3pt}
}

\node[module,text width=4.5cm] (relocalization)
  {3. Re-localization\\\relocalizationtext};
\node[module,text width=4.5cm,below=of relocalization] (tracking)
  {2. Tracking\\\trackingtext};
\node[module,text width=4.5cm,below=of tracking] (triangulation)
  {4. Triangulation\\\triangulationtext};
\node[module,text width=4.5cm,below=of triangulation] (mapping)
  {5. Mapping\\\mappingtext};
\node[yshift=-1.17cm, node distance=1.7cm, module, rotate=90, text width=5.05cm, left=of relocalization] (initialization)
  {1. Initialization};

\node[draw,dashed,rounded corners=8pt,inner sep=8pt,fit={(relocalization) (tracking) (triangulation) (mapping) (initialization)}]
  (fit) {};

\node[node distance=6cm, module,right=of tracking, text width=3.75cm,text width=2cm] (mekf)
  {EKF};

\node[module,text width=3.75cm, above=of mekf, xshift=0.85cm] (refgen)
  {Reference generator};
\node[module,below=of mekf,xshift=0.85cm,text width=3.75cm] (controller)
  {Proposition~\ref{prop:geomSUT}};
\node[module,below=of controller,text width=2cm,xshift=-0.85cm] (camera)
  {On-board\\camera};
\node[module,right=of camera,text width=1cm] (motorcontrol)
  {Motor\\control};
\node[module,left=of camera,text width=1cm] (imu)
  {IMU};

\node[draw,dashed,rounded corners=8pt, inner sep=8pt,fit={(refgen) (mekf) (controller) (imu) (motorcontrol)}]
  (fit) {};

\draw[myarr] (mekf.south) -- node[pos=0.5, left] {$\xvec$} ([yshift=-0.5cm]mekf.south);
\draw[myarr] ([yshift=0.5cm]motorcontrol.north) -- node[pos=0.5, left] {$\uvec$} (motorcontrol.north);

\draw[myarr] ([xshift=1.5cm]refgen.south) -- node[pos=0.88, left] {$(\xvec_r,\uvec_r)$} ([xshift=1.5cm]controller.north);

\draw[myarr, black] ([yshift=5pt]tracking.east) -- node[pos=0.41, above] {\small Pose estimates} ([yshift=5pt]mekf.west);
\draw[myarr, blue] (camera.south) -- ++(0,-2cm) -| node[pos=0.06, above] {\small Video feed} (initialization.west);

\draw[myarr, blue] ([yshift=-2cm,xshift=-6cm]camera.south) |- (relocalization.east);
\draw[myarr, blue] ([yshift=-2cm,xshift=-6cm]camera.south) |- ([yshift=-5pt]tracking.east);

\draw[myarr, gray] (imu.north) |- node[pos=0.35, above] {} (mekf.west);
\draw[myarr, gray] (imu.north) |- node[rotate=90,pos=0.24, above] {\small IMU data} (tracking.east);

\draw[myarr, black] ([xshift=2.5pt]relocalization.south) -- node[xshift=-7pt,pos=0.5, left] {\small Bad pose} ([xshift=2.5pt]tracking.north);
\draw[myarr, black] ([xshift=-2.5pt]tracking.north) -- node[xshift=5pt, pos=0.5, right] {\small Re-localized pose} ([xshift=-2.5pt]relocalization.south);

\draw[myarr, black] ([xshift=2.5pt]tracking.south) -- node[pos=0.5, right] {\small Pose Frames} ([xshift=2.5pt]triangulation.north);

\draw[myarr, blue] (tracking.west) --  ++(-0.7cm,0)   |- node[above, rotate=90, pos=0.24, right, yshift=0.25cm] {\small Pose Frames} (mapping.west);

\draw[myarr, black] ([xshift=2.5pt]triangulation.south) -- node[pos=0.5, right] {\small 3D point proposals} ([xshift=2.5pt]mapping.north);

\draw[myarr, gray] ([yshift=5pt]mapping.west) -- ++(-0.5cm,0)  |- node [below, pos=0.37, rotate=90]{\small Map} ([yshift=-5pt]tracking.west);
\draw[myarr, gray] ([yshift=2.5pt,xshift=-0.5cm]triangulation.west) --  ([yshift=2.5pt]triangulation.west);

\draw[myarr, black] ([yshift=-5pt, xshift=-1.2cm]mapping.west)  -- ([yshift=-5pt]mapping.west);
\draw[myarr, black] ([yshift=-2.5pt, xshift=-1.2cm]triangulation.west)  -- ([yshift=-2.5pt]triangulation.west);
\draw[myarr, black] ([yshift=5pt, xshift=-1.2cm]tracking.west)  -- ([yshift=5pt]tracking.west);
\end{tikzpicture}}
\caption{Holistic view of the control system modules with the main information flow. The SLAM system (left) is run on an external computer, and the modules implemented in the Crazyflie firmware (right) are run on its ARM processor. The communication between the dashed boxes is facilitated by two independent radios.}
\label{fig:systemoverview}
\end{figure}
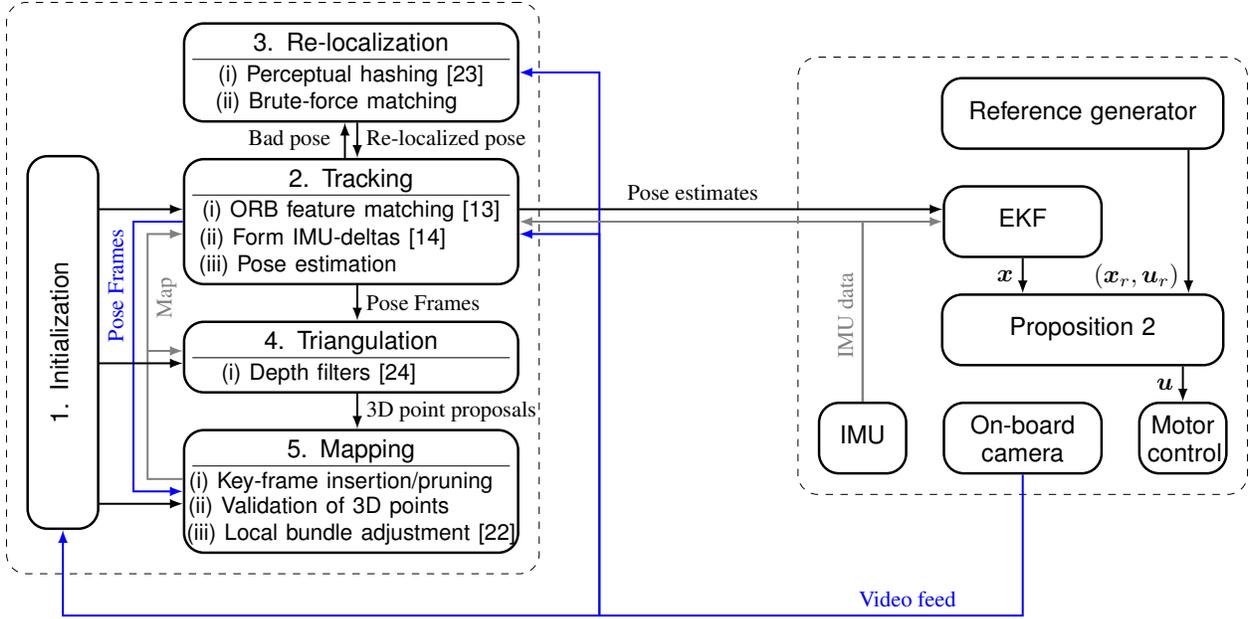}

\subsection{Initialization}\label{sec:SLAM:initialization}

The system has three phases: the initialization phase, the non-metric phase, and the metric phase (the metric phase being the operational phase). In the first phase an initial two-frame solution is found by selecting a reference frame and then as new frames arrive, estimating an essential matrix between the frames using the five point solver in \cite{five_point_solver} and the RANSAC framework in \cite{ransac}. From the essential matrix the pose and 3D structure is extracted, and if the set of 3D points that satisfy a re-projection threshold is sufficiently large and the median depth is sufficiently low (indicating adequate parallax), the solution is accepted and is then sent to the mapping module. This module in turn bundles the solution and sends it to the tracking module. If the initialization fails at any step the solution is rejected and the next frame is tried. If it fails too many times a new reference frame is selected.

This moves the system to the non-metric phase where the solution is defined up to similarity transformation, and no IMU data is used to constrain the scale and rotation of the system. Once the system has collected five keyframes (to be defined in Section \ref{sec:map}), the solution is upgraded to a metric solution using the method in \cite{map_reuse}, by solving for the gravity direction and scale. Now the system enters the operational phase, where the UAV can use the estimated pose.

\subsection{Tracking}\label{sec:SLAM:tracking}

The tracking module tracks the pose of consecutive frames relative to the current map. It uses the previous frame (and the IMU data when in the metric phase), to generate a proposition of the pose of the current frame. By using the 3D points seen in the previous frame, a new set of potentially visible 3D points are selected from the map that are co-visible with the previous points. These are projected into the current frame to facilitate a guided search for feature matches. The pose is then optimized using the matched features and in the case of metric phase, together with the past 3 poses and their corresponding IMU deltas. This process is repeated with the optimized pose and a smaller search window in the guided search. Next, matches with large re-projection error are discarded, and if the number of matches are sufficiently large and the pose is consistent with the IMU data, the pose is sent to the other modules. \inArxiv{If the tracking fails, a re-localization request is sent to the re-localization module, which then returns with the pose of the current frame if possible.}

\subsection{Re-localization}\label{sec:SLAM:relocalization}

The re-localization module uses perceptual hashing in \cite{Zauner2010ImplementationAB} to turn images into hashes, where similar images have similar hashes. When performing re-localization of a query frame, the frame is converted to a hash. The hash can then efficiently be compared to the hashes of the keyframes in the map to find frames that have observed the same view. The most similar frames are then selected as candidates. For each such frame, the observed 3D points are matched to the query frame features using brute-force matching together with a three point pose solver in \cite{pose_ransac} and RANSAC in \cite{ransac}. Any camera pose with a sufficient amount of inlier matches is then sent back to the tracking module as a candidate pose.

\subsection{Triangulation}\label{sec:SLAM:triangulation}

To triangulate new 3D points and extend the current map, a two-step approach is used. Firstly, for each keyframe, the inverse depth filters in \cite{VOGIATZIS2011434} are used to features that have not already been associated with 3D points. The filter maintains a Gaussian distribution over inverse depth of the 3D point from the keyframe's point of view, as well as a beta distribution over measurement update inlier ratio. For each new pose frame sent from the tracker, the co-visible filters are updated. If the inlier ratio of a filter becomes lower than a threshold, it is restarted. If the uncertainty of the inverse depth becomes sufficiently low, the inverse depth is considered to have converged and is sent from the triangulation module to the mapping module. In the next step, the inverse depth is converted to a 3D point distribution using the inverse depth uncertainty from the filter and an assumed pixel noise of one pixel. This distribution is then projected into co-visible keyframes and features close to the projection, with a sufficiently low descriptor error, are considered to be potential matches. In the next step, a two-point DLT based triangulation RANSAC procedure in~\cite{hartley1997triangulation} is used to find the 3D point with the largest inlier set. If this 3D point has been seen by four keyframes, then it is accepted and is given a life of 5. Each time the 3D point is seen by the tracker, the life is increased. Each time it is predicted but not seen, its life is decreased. If the life reaches zero, the 3D point is removed.


\subsection{Mapping}\label{sec:map}

The mapping module is responsible of updating, optimizing and pruning the map, as well as sending the changes to the rest of the modules.  Using all of the frames in the video would quickly make the optimization problem prohibitively large and additionally many of the frames would be near identical and unable to provide much information. Therefore, a small subset is selected, here called keyframes, that accurately represent the solution. The map object contains the 3D points, the IMU-deltas, and the keyframes.

Keyframes are added in real-time as the map expands, and must satisfy certain conditions to be added: (i) compared to the closest keyframe, the camera must have moved at least 2.5\% of the current mean depth; (ii) the mutual overlap with the most overlapping keyframe must be less than 90 \%, or the average uncertainty of the projection of the 3D points must be higher than four pixels; and (iii) the number of tracked points must be larger than ten. This leads to a generous keyframe insertion policy, quickly expanding the 3D point set, but introducing redundancy in the keyframes.

A redundant keyframe can be characterized as observing the same 3D points as many of the other keyframes. To detect and remove redundant keyframes, we keep track of the amount of 3D points each keyframe observes that have also been observed by at least seven other keyframes. If the fraction of these 3D points is higher than 90\%, the keyframe is marked as redundant. From the set of redundant keyframes, the frame that is the closest to another keyframe is removed.

Every time a new keyframe is added, a local bundle adjustment is performed where keyframes and 3D points that are co-visible with the new frame are adjusted and the remaining are kept fixed. This is done to prevent the map from diverging, and only a local portion of the map is optimized to keep the computational time bounded. Whenever the map has been updated, either by adding a new keyframe, by adding a new 3D point, or by performing a pruning, the difference between the previous map and the updated map is extracted and is then sent to the rest of the modules.

\subsection{Integration}\label{sec:SLAM:integration}

To fuse SLAM estimates with the IMU-data and generate a full state estimate of the UAV given the dynamics in~\eqref{eq:FSFFull:contsys}, we consider the IMU-driven multiplicative extended Kalman filter (EKF) proposed in~\cite{mueller2015fusing}, with a first order attitude reset in~\cite{MuellerCovariance2016}. The filter assumes a non-linear UAV model corresponding to the dynamics in~\eqref{eq:FSFFull:contsys}, but with the velocities $\vvec$ expressed in $\{\Bcal\}$, and the attitude parameterized as a first-order attitude error, $\deltabf\in\Real^3$. This attitude error relates to the estimate $\hat\R(t)\in\SOT$ as $\hat\R(t)=\hat\R(t_k)(\I+\Sbf(\deltabf(t)))$, and is reset to zero when it exceeds a predefined threshold with $t_k$ denoting  the most recent  reset. As such, the state of the EKF is given by $\zvec = (\pvec;\;\vvec;\;\deltabf)\in\Real^9$, and the dynamics in~\eqref{eq:FSFFull:contsys} are expressed in $\zvec$. This estimate is subsequently externalized into the full state of the UAV, $\xvec=(\pvec, \vvec, \Xbf, \omegabf)$, where the attitude rates are computed by averaging the gyroscopic measurements between each EKF prediction. In the experiments, a scalar update version of the filter is used (see, e.g.,~\protect{\cite[Chapter 6]{greiff2017modelling}}), incorporating the measurements consecutively as they arrive. For additional details on the filter implementation, refer to~\cite{mueller2015fusing,MuellerCovariance2016,crazyflie2021mekf}.

\inArxiv{The IMU-data from the Crazyflie are transferred to a PC using a version of the robot operating system (ROS) driver in~\cite{crazyflieROS}, and the video is streamed over a different radio and processed directly by the SLAM system. This stack subsequently outputs an estimate of the UAV pose at 50 [Hz], which is communicated back to the UAV via radio through the ROS driver. The positional part of the estimate is queued into the EKF, which fuses the positional information with the IMU measurements at a rate of 100 [Hz]. The reference trajectory is also communicated to the UAV via ROS as a set of splines in the flat output space of the UAV, and the expansion of this trajectory into the states of the reference dynamics in~\eqref{eq:FSFFull:contsys} and the controller in Proposition~\ref{prop:geomSUT}, both run at a rate of 500 [Hz] (see Figure~\ref{fig:systemoverview}).}

\section{NUMERICAL EXAMPLES AND EXPERIMENTS}\label{sec:numerical}
In this section, we start by giving a simulation example in Section~\ref{sec:simexample} demonstrating the tracking properties of the controller in Proposition~\ref{prop:geomSUT} for a circular maneuver with large initial errors. We then present an experiment in Section~\ref{sec:rtexample} where the controller is integrated with the SLAM system, while also removing the desired attitude accelerations (by letting $\dot{\omegabf}_d\triangleq 0$) to avoid differentiating the control errors.

\subsection{Simulation Example}\label{sec:simexample}
As the proposed controller is asymptotically attractive and ULES for all parameters $m>0,\J=\J\Tr\succ\Z$ and all initial errors, we take $m=0.1$ with $g=10$, and sample a random inertia matrix (here chosen such that $\mineig(\J)=0.05$ and $\maxeig(\J)=0.1$). The initial conditions of the system are sampled from $\vvec(\tz),\omegabf(\tz)\sim \Ncal(\Z,5\I)$, with $\pvec(\tz)\sim\Ncal((0,0,-2)\Tr,\I)$, and $\R(\tz)\sim\Ucal(\SOT)$. From this random initial state, the system is controlled along
\begin{equation*}
\pvec_r(t)=(3\sin(t),3\cos(t),0)\Tr, \quad \bvec_{1r}(t)=\vvec_{r}(t)/\|\vvec_{r}(t)\|.
\end{equation*}
This trajectory can be expanded into the full state trajectory of the UAV in~\eqref{eq:FSFFull:contsys} using the property of differential flatness (see, e.g.,~\protect{\cite[Chapter 3]{greiff2017modelling}}). In one particular realization,
\inACC{\begin{equation*}
\R(\tz) \hspace{-2pt}=\hspace{-2pt}
\begin{bmatrix}
 0.51  \hspace{-2pt}&\hspace{-2pt} -0.05  \hspace{-2pt}&\hspace{-2pt} -0.86\\
   -0.78 \hspace{-2pt}&\hspace{-2pt}   0.41  \hspace{-2pt}&\hspace{-2pt} -0.48\\
    0.37 \hspace{-2pt}&\hspace{-2pt}   0.91 \hspace{-2pt}&\hspace{-2pt}   0.17\\
\end{bmatrix},\;\;
\J \hspace{-2pt}=
\begin{bmatrix}
    0.08 \hspace{-2pt}&\hspace{-2pt}   0.01 \hspace{-2pt}&\hspace{-2pt}   0.02\\
    0.01 \hspace{-2pt}&\hspace{-2pt}   0.07 \hspace{-2pt}&\hspace{-2pt}   0.01\\
    0.02 \hspace{-2pt}&\hspace{-2pt}   0.01 \hspace{-2pt}&\hspace{-2pt}   0.07
\end{bmatrix}\hspace{-2pt},
\end{equation*}
\begin{equation}\label{eq:initialsim}
\hspace{-1pt}\pvec(\tz) \hspace{-2pt}=\hspace{-2pt}
\begin{bmatrix}
    0.08\\
   -0.16\\
   -1.63
\end{bmatrix}\hspace{-2pt},\;
\vvec(\tz) \hspace{-2pt}=\hspace{-2pt}
\begin{bmatrix}
   -0.59\\
    0.76\\
   -0.95
\end{bmatrix}\hspace{-2pt},\;
\omegabf(\tz) \hspace{-2pt}=\hspace{-2pt}
\begin{bmatrix}
   -1.81\\
    1.80\\
    2.81
\end{bmatrix}\hspace{-4pt}.\hspace{-7pt}
\end{equation}}
\inArxiv{\begin{align}
\R(\tz) &=
\begin{bmatrix}
 0.51  & -0.05  & -0.86\\
   -0.78 &   0.41  & -0.48\\
    0.37 &   0.91 &   0.17\\
\end{bmatrix},\;\;
\J =
\begin{bmatrix}
    0.08 &   0.01 &   0.02\\
    0.01 &   0.07 &   0.01\\
    0.02 &   0.01 &   0.07
\end{bmatrix},\\
\pvec(\tz)& =
\begin{bmatrix}
    0.08\\
   -0.16\\
   -1.63
\end{bmatrix},\;
\vvec(\tz) =
\begin{bmatrix}
   -0.59\\
    0.76\\
   -0.95
\end{bmatrix},\;
\omegabf(\tz) =
\begin{bmatrix}
   -1.81\\
    1.80\\
    2.81
\end{bmatrix}.\label{eq:initialsim}
\end{align}}

For this realization of the UAV parameters and the initial errors, the control signals, reference trajectory and system response are depicted in~Fig.~\ref{fig:ACC_sim_A}. Here, we note that the Lyapunov function in~\eqref{eq:FSF:full:SU2:lyap} quickly decays to a small value (here shown in the 10-logarithm), and that it satisfies the associated quadratic bounds at all times. This holds despite the system being initialized outside of $D$. As such, we here rely on the (almost global) asymptotic attractiveness properties in Proposition~\eqref{prop:FSFFull:attractivitySU2} before reaching the domain of exponential attraction. The system configurations are depicted in time in Fig.~\ref{fig:ACC_sim_B}\inArxiv{, with color coding of $\{\Bcal\}$ corresponding to Fig.~\ref{fig:geometry}}. Given that a new set of initial conditions and system parameters can be sampled, similar convergence properties were verified in a total of $10^3$ realizations of the parameters and initial errors in~\eqref{eq:initialsim}. 

\begin{figure}[h!]
    \centering
    \if\useieeelayout1
    \includegraphics[width=\columnwidth]{figures/ACC_sim_A.png}\vspace{-10pt}
    \else
    \includegraphics[width=0.9\textwidth]{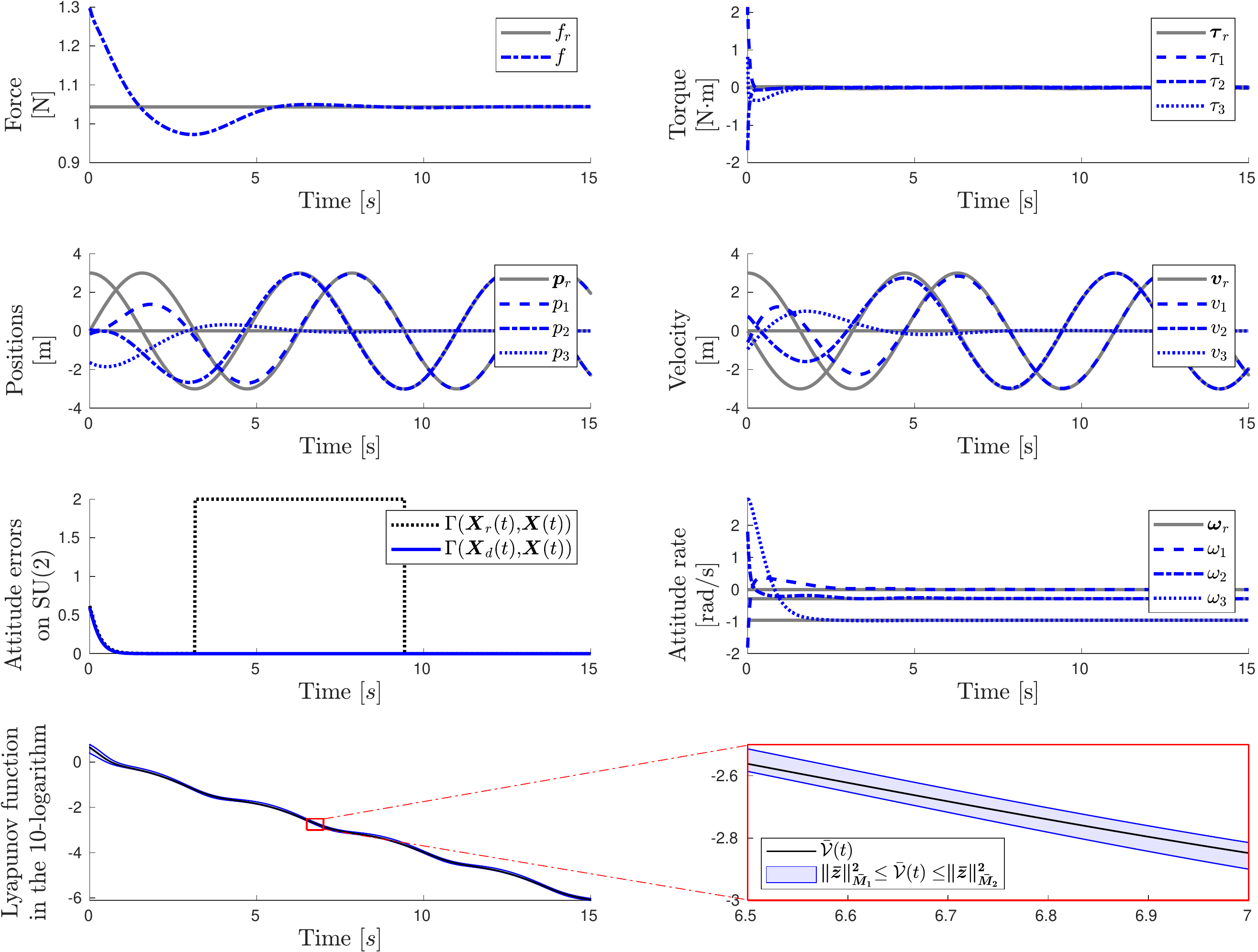}
    \fi
    \caption{System response and control signals in time with references in gray and signals of the controlled system in blue. \textit{Top, left}: Reference- and  controlled force.
    \textit{Top, right}: Reference- and controlled torque.
    \textit{Top center, left}: Reference- and controlled position.
    \textit{Top center, right}: Reference- and controlled velocity.
    \textit{Bottom center, left}: Distance to the reference attitude (black) and desired attitude (blue).
    \textit{Bottom center, right}: Reference- and controlled attitude rate.
    \textit{Bottom}: The Lyapunov function depicted with the associated quadratic bounds in blue, depicted over $t\in[0,15]$ to the left, with a zoom indicated in red on $t\in[6.5,7]$ to the right.
    }
    \label{fig:ACC_sim_A}
\end{figure}

\begin{figure}[h!]
    \centering
    \if\useieeelayout1
    \includegraphics[width=\columnwidth]{figures/ACC_sim_B.png}\vspace{-10pt}
    \else
    \includegraphics[width=0.9\textwidth]{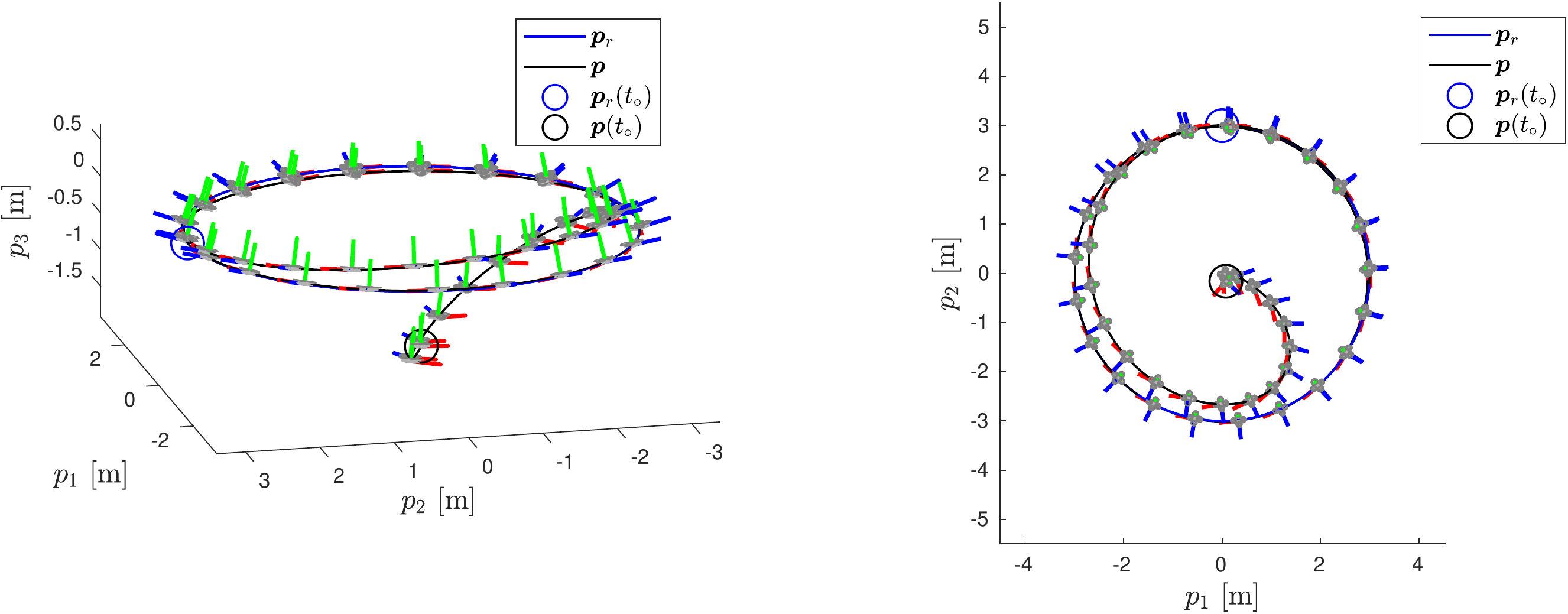}
    \fi
    \caption{Configurations of the UAV $(\pvec(t),\Xbf(t))\in\Real^3\times \SUT$ when controlled \inACC{along the trajectory} using Proposition~\ref{prop:geomSUT}.}
    \label{fig:ACC_sim_B}
\end{figure}

\subsection{Real-Time Example}\label{sec:rtexample}
In the second example, an inventorying experiment is conducted in real-time with a Crazyflie 2.0~\cite{crazyflie2021twopointo} with respect to a set of shelves in the coroner of a room \inACC{(see Fig.~\ref{fig:dronelab})}\inArxiv{(see Fig.~\ref{fig:geometry})}. The SLAM system is integrated with the stock multiplicative EKF of the Crazyflie~\cite{crazyflie2021mekf} as described in Section~\ref{sec:estimation}, using the positional estimates from the SLAM system in combination with the IMU-measurements to generate a full-state estimate. These estimates are subsequently used by the controller in Proposition~\ref{prop:geomSUT} to actuate the UAV, but employing the discontinuous version of the attitude controller in~\cite{greiff2021similarities}.

\inACC{\begin{figure}[h!]
    \centering
    \if\useieeelayout1
    \includegraphics[width=\columnwidth]{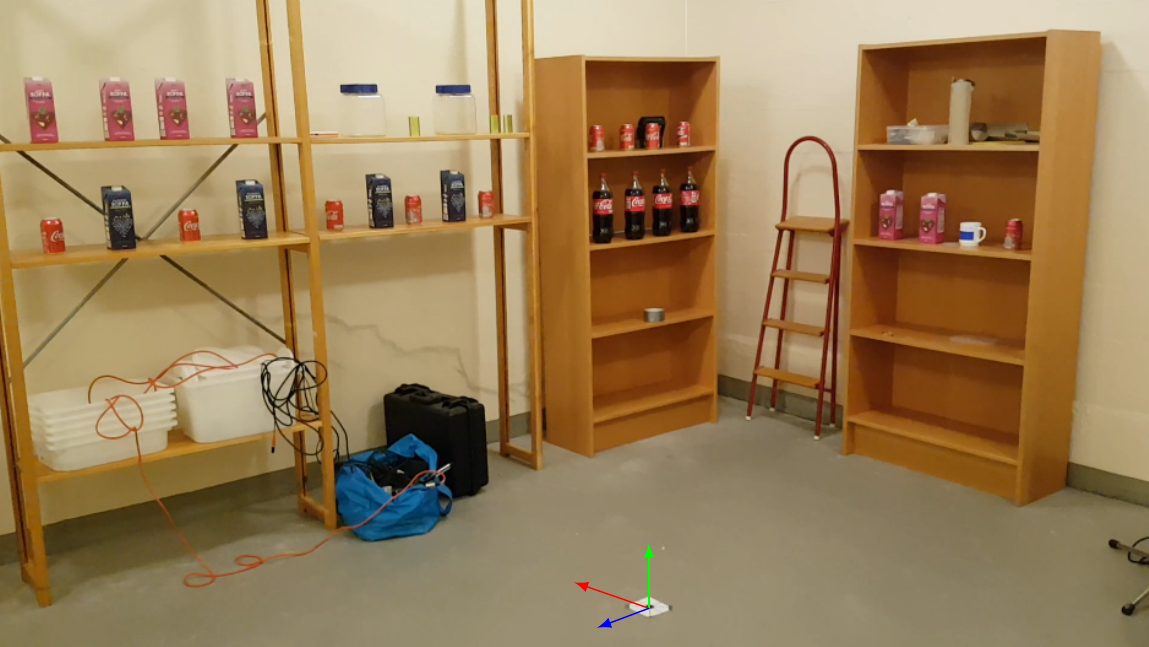}\vspace{-10pt}
    \else
    \includegraphics[width=0.7\textwidth]{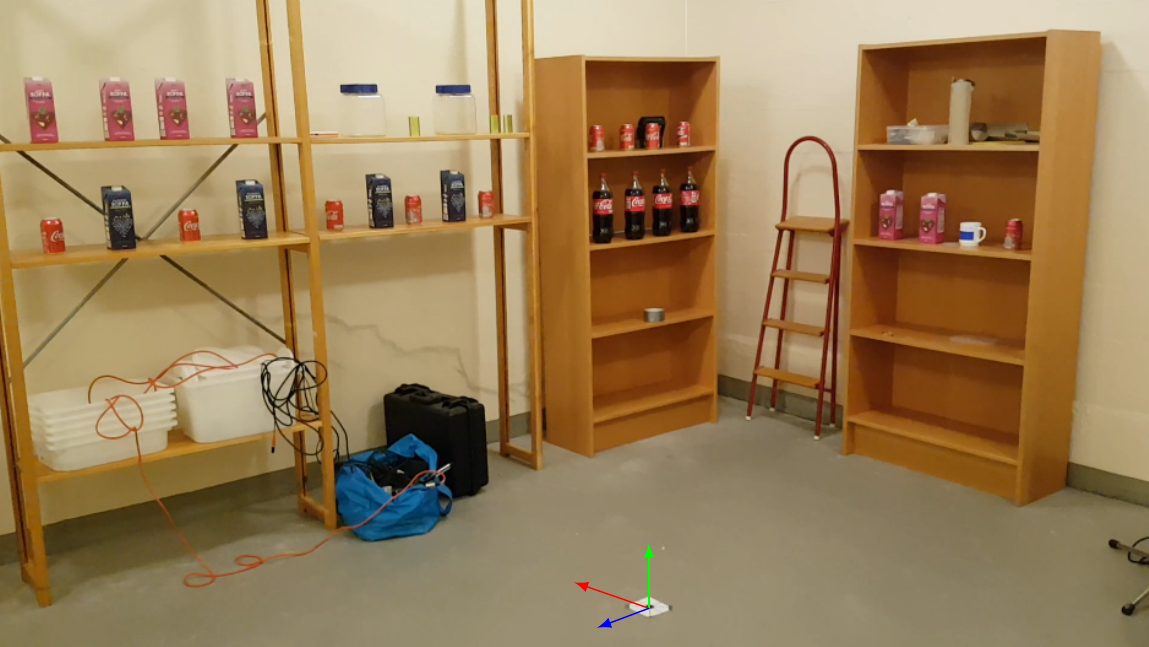}
    \fi
    \caption{Shelves with drinks are to be scanned, with the camera facing the target shelf at all times. The origin of $\{\Gcal\}$ is defined on a piece of paper.}
    \label{fig:dronelab}
\end{figure}}

The shelf geometry is assumed to be known, and a reference trajectory is planned consisting of linear splines such that the UAV traverses each segment at a velocity of $\|\vvec_r(t)\|=1$ [m/s] approximately 0.4 [m] from the shelves. In order to keep the shelves within camera view, this implies performing a turn at $t\in[24.6,26.3]$, also defined by linear splines. As such, the reference trajectory cannot be followed perfectly at the spline end-points, where the velocities are discontinuous, resulting in slight transients in the errors. The tracking errors are shown in Fig.~\ref{fig:ACC_rtexp}, with the expanded reference trajectory and positional estimates logged from the UAV, and the attitude error computed between the estimated rotation in the EKF and the rotation in the SLAM system (which is not incorporated in the EKF). To highlight that the attitude error does not correspond to the distance between the estimate of the EKF and the desired attitude (i.e., the control error $\Xbf_d\Hr\Xbf$ in the UAV), the attitude error is depicted in the $\Psi$-distance, with $\R_r$ sampled from the UAV reference trajectory and $\R$ estimated from the SLAM system.

\begin{figure}
    \centering
    \if\useieeelayout1
    \includegraphics[width=\columnwidth]{figures/ACC_RT.png}\vspace{-10pt}
    \else
    \includegraphics[width=0.9\textwidth]{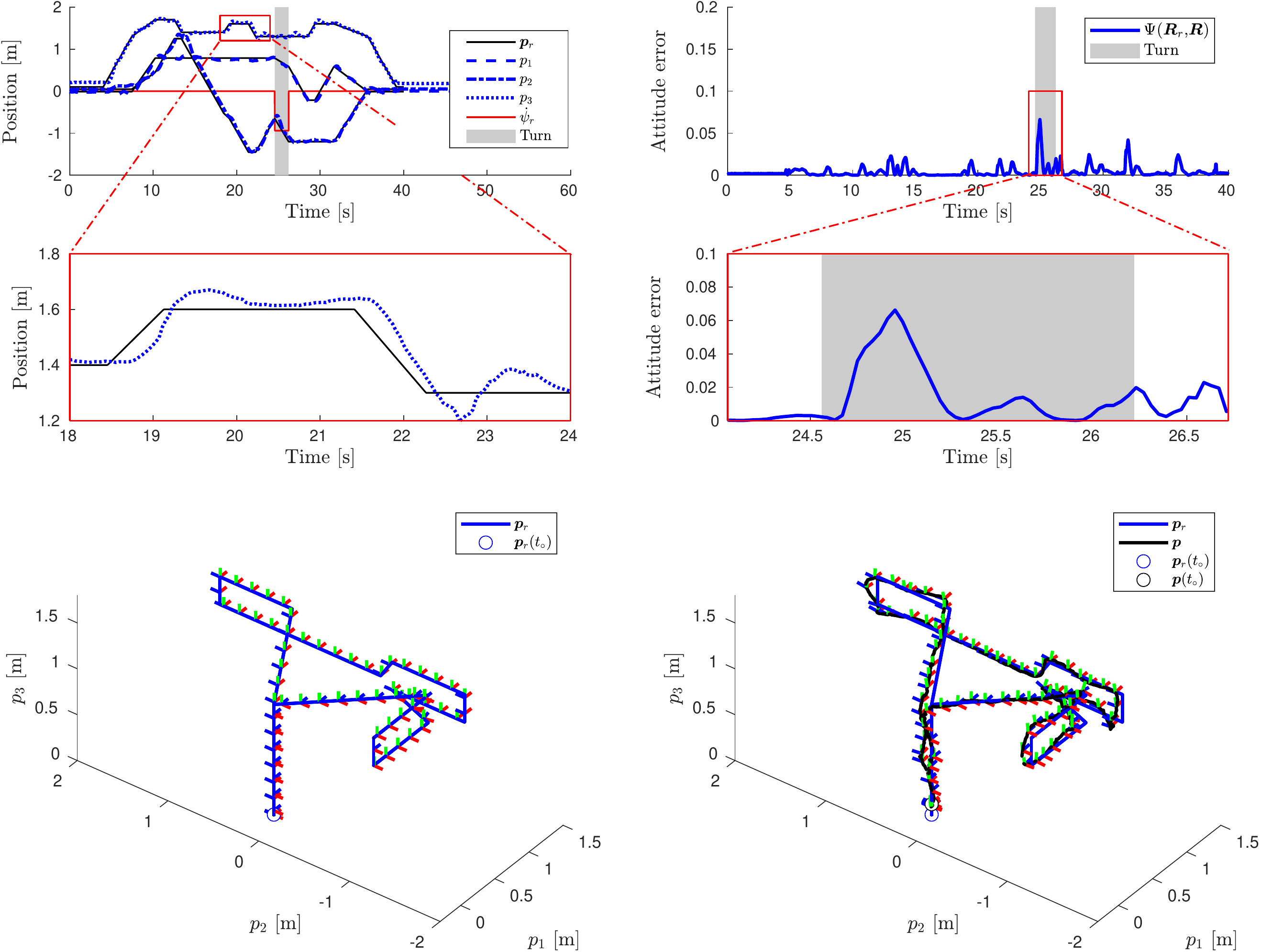}
    \fi
\caption{\textit{Top, left:} The positional reference $\pvec_r$ (black) and the response $\pvec$ (blue), along with the time-derivative of the yaw angle reference (red) indicating the time-interval during which the system is turning (gray). \textit{Top, right:} The tracking attitude error in the distance on $\SOT$, slightly increasing during the turn, but otherwise small. \textit{Center, left:} Zoom on the positional response in the elevation just before the turn, showing slight delays in the response and an overshoot.
\textit{Center, right:} Zoom on the attitude error during the turn.
\textit{Bottom, left:} Reference configurations trajectory in $\Real^3$ in time.
\textit{Bottom, left:} System response in $\Real^3$ in time, showing the reference trajectory, positional response, and measured rotation. }
    \label{fig:ACC_rtexp}
\end{figure}

The UAV successfully scans the shelves in rapid succession, following the reference trajectory down to the expected tracking errors induced by a lack of continuity when switching between splines. We also note that the attitude error is relatively small throughout the experiment, with slight increases when switching between splines, and a larger error during the start of the turning maneuver. Again, we emphasize that this is the attitude error between the reference trajectory and the rotation estimated in the SLAM system, which is not explicitly used in the controller. Finally, to get a sense of the accuracy of the slam system, the initial and terminal configuration of the UAV is depicted in Fig.~\ref{fig:UAVconfigurations}. This demonstrates that without any external motion capture, the controlled system navigates back to a point that differs from the initial position by a few centimeters, despite facing a different wall while landing. Supporting videos of the simulation as well as the experiment are published in~\cite{videoInventorying}.

\section{CONCLUSIONS}\label{sec:conclusion}
In this paper, we have presented a geometric tracking controller for quadrotor UAVs configured on $\SUT\times\Real^3$, that is analogous to the geometric tracking controller on $\SOT\times \Real^3$ in~\cite{lee2010geometric} -- yet distinctly different in ways that have meaningful consequences. In addition, a SLAM system was implemented based on ORB features to process monocular video to a sequence of a pose estimates that were subsequently fused with IMU-data in an on-board EKF. The controller in Proposition~\ref{prop:geomSUT} was demonstrated in simulation, before being applied in practice to a real-time inventorying scenario. We emphasize that the proposed control system is capable of actuating the UAV along the desired reference trajectory without any external motion capture system. As such, a UAV equipped with this control system is easily deployed and is a low-cost alternative for supermarket inventorying.

\bibliography{references}{}

\begin{thebibliography}{10}
\providecommand{\url}[1]{#1}
\csname url@rmstyle\endcsname
\providecommand{\newblock}{\relax}
\providecommand{\bibinfo}[2]{#2}
\providecommand\BIBentrySTDinterwordspacing{\spaceskip=0pt\relax}
\providecommand\BIBentryALTinterwordstretchfactor{4}
\providecommand\BIBentryALTinterwordspacing{\spaceskip=\fontdimen2\font plus
\BIBentryALTinterwordstretchfactor\fontdimen3\font minus
  \fontdimen4\font\relax}
\providecommand\BIBforeignlanguage[2]{{%
\expandafter\ifx\csname l@#1\endcsname\relax
\typeout{** WARNING: IEEEtran.bst: No hyphenation pattern has been}%
\typeout{** loaded for the language `#1'. Using the pattern for}%
\typeout{** the default language instead.}%
\else
\language=\csname l@#1\endcsname
\fi
#2}}

\bibitem{skydio2021}
\BIBentryALTinterwordspacing
Skydio, ``Skydio product homepage,'' last accessed at 2021-09-01. [Online].
  Available: \url{https://www.skydio.com/}
\BIBentrySTDinterwordspacing

\bibitem{fernandez2019towards}
T.~M. Fern{\'a}ndez-Caram{\'e}s, O.~Blanco-Novoa, I.~Froiz-M{\'\i}guez, and
  P.~Fraga-Lamas, ``Towards an autonomous industry 4.0 warehouse: A {UAV} and
  blockchain-based system for inventory and traceability applications in big
  data-driven supply chain management,'' \emph{Sensors}, vol.~19, no.~10, p.
  2394, 2019.

\bibitem{fresk2013full}
E.~Fresk and G.~Nikolakopoulos, ``\refok{{Full quaternion based attitude
  control for a quadrotor}},'' in \emph{2013 European Control Conference
  (ECC)}.\hskip 1em plus 0.5em minus 0.4em\relax IEEE, 2013, pp. 3864--3869.

\bibitem{brescianini2013nonlinear}
D.~Brescianini, M.~Hehn, and R.~D'Andrea, ``\refok{{N}onlinear quadrocopter
  attitude control: Technical report},'' ETH Zurich, Tech. Rep., 2013.

\bibitem{brescianini2018tilt}
D.~Brescianini and R.~D’Andrea, ``\refok{{T}ilt-prioritized quadrocopter
  attitude control},'' \emph{IEEE Transactions on Control Systems Technology},
  vol.~28, no.~2, pp. 376--387, 2018.

\bibitem{greiff2021similarities}
M.~Greiff, Z.~Sun, and A.~Robertsson, ``\refok{{A}ttitude Control on {SU(2)}:
  Stability, Robustness, and Similarities},'' \emph{IEEE Control Systems
  Letters}, vol.~6, pp. 73--78, 2021.

\bibitem{lee2010geometric}
T.~Lee, M.~Leok, and N.~H. McClamroch, ``\refok{{G}eometric tracking control of
  a quadrotor {UAV} on {SE(3)}},'' in \emph{49th IEEE Conference on Decision
  and Control (CDC)}.\hskip 1em plus 0.5em minus 0.4em\relax IEEE, 2010, pp.
  5420--5425.

\bibitem{goodarzi2013geometric}
F.~Goodarzi, D.~Lee, and T.~Lee, ``\refok{{G}eometric nonlinear {PID} control
  of a quadrotor {UAV} on {SE(3)}},'' \emph{2013 European control conference
  (ECC)}, pp. 3845--3850, 2013.

\bibitem{lee2013nonlinear}
T.~Lee, M.~Leok, and N.~H. McClamroch, ``\refok{{N}onlinear robust tracking
  control of a quadrotor {UAV} on {SE(3)}},'' \emph{Asian Journal of Control},
  vol.~15, no.~2, pp. 391--408, 2013.

\bibitem{lee2015global}
T.~Lee, ``\refok{{G}lobal Exponential Attitude Tracking Controls on {SO(3)}},''
  \emph{IEEE Transactions on Automatic Control}, vol.~60, no.~10, pp.
  2837--2842, 2015.

\bibitem{kaufmann2020deep}
\BIBentryALTinterwordspacing
E.~Kaufmann, A.~Loquercio, R.~Ranftl, M.~Müller, V.~Koltun, and D.~Scaramuzza,
  ``\refok{{D}eep Drone Acrobatics},'' 2020. [Online]. Available:
  \url{arxiv.org/abs/2006.05768}
\BIBentrySTDinterwordspacing

\bibitem{falanga2018pampc}
D.~Falanga, P.~Foehn, P.~Lu, and D.~Scaramuzza, ``{PAMPC: Perception-aware
  model predictive control for quadrotors},'' in \emph{2018 IEEE/RSJ
  International Conference on Intelligent Robots and Systems (IROS)}.\hskip 1em
  plus 0.5em minus 0.4em\relax IEEE, 2018, pp. 1--8.

\bibitem{orb}
E.~Rublee, V.~Rabaud, K.~Konolige, and G.~Bradski, ``{ORB: An efficient
  alternative to SIFT or SURF},'' in \emph{2011 International Conference on
  Computer Vision}, 2011, pp. 2564--2571.

\bibitem{preintegration}
C.~Forster, L.~Carlone, F.~Dellaert, and D.~Scaramuzza, ``On-manifold
  preintegration for real-time visual--inertial odometry,'' \emph{IEEE
  Transactions on Robotics}, vol.~33, no.~1, pp. 1--21, 2017.

\bibitem{crazyflie2021twopointo}
\BIBentryALTinterwordspacing
Bitcraze, ``\refok{{C}razyflie 2.0},'' 2021, last accessed: 10-4-2021.
  [Online]. Available:
  \url{www.bitcraze.io/products/old-products/crazyflie-2-0}
\BIBentrySTDinterwordspacing

\bibitem{lee2010control}
\BIBentryALTinterwordspacing
T.~Lee, M.~Leok, and N.~H. McClamroch, ``\refok{{C}ontrol of complex maneuvers
  for a quadrotor {UAV} using geometric methods on {SE(3)}},'' 2011. [Online].
  Available: \url{arxiv.org/abs/1003.2005}
\BIBentrySTDinterwordspacing

\bibitem{hall2015lie}
B.~Hall, \emph{\refok{{L}ie groups, Lie algebras, and representations: an
  elementary introduction}}.\hskip 1em plus 0.5em minus 0.4em\relax Springer
  International Publishing, Switzerland, 2015, vol. 222.

\bibitem{khalil2002nonlinear}
H.~Khalil, \emph{\refok{{N}onlinear systems}}, 3rd~ed.\hskip 1em plus 0.5em
  minus 0.4em\relax Prentice hall Upper Saddle River, New Jersey, USA, 2002.

\bibitem{mayhew2011quaternion}
C.~G. Mayhew, R.~G. Sanfelice, and A.~R. Teel, ``\refok{{O}n quaternion-based
  attitude control and the unwinding phenomenon},'' in \emph{Proceedings of the
  2011 American Control Conference}.\hskip 1em plus 0.5em minus 0.4em\relax
  IEEE, 2011, pp. 299--304.

\bibitem{mayhew2012path}
------, ``\refok{{O}n path-lifting mechanisms and unwinding in quaternion-based
  attitude control},'' \emph{IEEE Transactions on Automatic Control}, vol.~58,
  no.~5, pp. 1179--1191, 2012.

\bibitem{greiff2018nonsingular}
M.~{Greiff} and A.~{Robertsson}, ``\refok{{I}ncremental Reference Generation
  for Nonsingular Control on {SE(3)}},'' in \emph{2018 IEEE Conference on
  Control Technology and Applications (CCTA)}, 2018, pp. 132--137.

\bibitem{isam2}
M.~Kaess, H.~Johannsson, R.~Roberts, V.~Ila, J.~Leonard, and F.~Dellaert,
  ``{iSAM2: Incremental smoothing and mapping with fluid relinearization and
  incremental variable reordering},'' in \emph{2011 IEEE International
  Conference on Robotics and Automation}, 2011, pp. 3281--3288.

\bibitem{Zauner2010ImplementationAB}
C.~Zauner, ``Implementation and benchmarking of perceptual image hash
  functions,'' Ph.D. dissertation, 2010, {ISBN}: 1446144429.

\bibitem{VOGIATZIS2011434}
G.~Vogiatzis and C.~Hern{\'a}ndez, ``Video-based, real-time multi-view
  stereo,'' \emph{Image and Vision Computing}, vol.~29, no.~7, pp. 434--441,
  2011.

\bibitem{five_point_solver}
D.~Nister, ``An efficient solution to the five-point relative pose problem,''
  \emph{IEEE Transactions on Pattern Analysis and Machine Intelligence},
  vol.~26, no.~6, pp. 756--770, 2004.

\bibitem{ransac}
M.~Fischler and R.~Bolles, ``Random sample consensus: a paradigm for model
  fitting with applications to image analysis and automated cartography,''
  \emph{Communications of the Association for Computing Machinery (ACM)},
  vol.~24, pp. 381--395, 1981.

\bibitem{map_reuse}
R.~Mur-Artal and J.~D. Tardós, ``{Visual-Inertial Monocular SLAM With Map
  Reuse},'' \emph{IEEE Robotics and Automation Letters}, vol.~2, no.~2, pp.
  796--803, 2017.

\bibitem{pose_ransac}
X.-S. Gao, X.-R. Hou, J.~Tang, and H.-F. Cheng, ``Complete solution
  classification for the perspective-three-point problem,'' \emph{IEEE
  Transactions on Pattern Analysis and Machine Intelligence}, vol.~25, no.~8,
  pp. 930--943, 2003.

\bibitem{hartley1997triangulation}
R.~I. Hartley and P.~Sturm, ``Triangulation,'' \emph{Computer Vision and Image
  Understanding}, vol.~68, no.~2, pp. 146--157, 1997.

\bibitem{mueller2015fusing}
M.~W. Mueller, M.~Hamer, and R.~D'Andrea, ``\refok{{F}using ultra-wideband
  range measurements with accelerometers and rate gyroscopes for quadrocopter
  state estimation},'' in \emph{2015 IEEE International Conference on Robotics
  and Automation (ICRA)}.\hskip 1em plus 0.5em minus 0.4em\relax IEEE, 2015,
  pp. 1730--1736.

\bibitem{MuellerCovariance2016}
M.~W. Mueller, M.~Hehn, and R.~D’Andrea, ``\refok{{C}ovariance Correction
  Step for {K}alman Filtering with an Attitude},'' \emph{Journal of Guidance,
  Control, and Dynamics}, vol.~40, no.~9, pp. 2301--2306, 2017.

\bibitem{greiff2017modelling}
M.~Greiff, ``\refok{{Modelling and control of the {C}razyflie quadrotor for
  aggressive and autonomous flight by optical flow driven state estimation}},''
  Master's thesis, Lund University, 2017.

\bibitem{crazyflie2021mekf}
\BIBentryALTinterwordspacing
Bitcraze, ``\refok{{MEKF} in the {C}razyflie},'' 2021, last accessed:
  10-4-2021. [Online]. Available:
  \url{github.com/bitcraze/crazyflie-firmware/blob/master/src/modules/src/kalman\_core/kalman\_core.c}
\BIBentrySTDinterwordspacing

\bibitem{crazyflieROS}
W.~H{\"o}nig and N.~Ayanian, \emph{Flying Multiple UAVs Using ROS}.\hskip 1em
  plus 0.5em minus 0.4em\relax Springer International Publishing, 2017, pp.
  83--118.

\bibitem{videoInventorying}
\BIBentryALTinterwordspacing
M.~Greiff, ``Inventorying with the crazyflie 2.0,'' last accessed at
  2021-10-03. [Online]. Available: \url{https://youtu.be/JIZ2KM-rYjk}
\BIBentrySTDinterwordspacing

\end{thebibliography}
\bibliographystyle{IEEEtran}

\inArxiv{
\cleardoublepage\appendix

\section*{Appendix}
\begin{proof}[Proof of Proposition~\ref{prop:geomSUT}]
This proof is strikingly similar to the proof for the controller on $\SOT\times \Real^3$ in~\protect{\cite[Appendix D]{lee2010geometric}}, and a sketch is given here with some key intermediary expressions. Define a Lyapunov function comprised of two parts, one relating the attitude subsystem, $\Lyap_a$, and one relating to the translation subsystem, $\Lyap_p$. The former is defined as in~\eqref{eq:Lyapufunccand}, and the latter is defined analogously, as
\begin{subequations}
\begin{align}
\Lyap^a&\defeq k_X\Gamma(\Xbf_d, \Xbf) + c_a\evec_{\omegabf}\cdot \evec_{\Xbf} + \frac{1}{2}\evec_{\omegabf}\cdot \J\evec_{\omegabf},\\
\Lyap^p&\defeq  \frac{1}{2}k_p\|\evec_{\pvec}\|^2 + \frac{1}{2}m\|\evec_{\vvec}\| ^2 + c_p \evec_{\pvec}\cdot \evec_{\vvec},
\end{align}
\end{subequations}
respectively. As the attitude dynamics are actuated along this trajectory by Proposition~\ref{thm:attitudeSUTcont}, consider initial attitude errors on
\begin{align}
(\Xbf_e(\tz),\evec_{\omegabf}(\tz))&\in \{(\Xbf_e,\evec_{\omegabf})\in\Lcal_{\phi}\times \Real^3\;|\;\Lyap^a|_{c_a=0}\leq k_X\phi\},
\end{align}
where
\begin{align}
\Lcal_{\phi} &= \{\Xbf_d\Hr\Xbf\in\SUT\;|\;\Gamma(\Xbf_d, \Xbf)\leq \phi<2\}.
\end{align}
From the result in Proposition~\ref{thm:attitudeSUTcont}, it follows that $\Xbf_e(t)\in\Lcal_{\phi}$ for all $t\geq \tz$ if $\Lyap^a(\tz)|_{c_a=0}\leq k_X\phi$, and that the errors converge exponentially to $(\Xbf_e,\evec_{\omegabf})=(\I,\Z)$. As such, the main idea of the proof is to conduct the stability analysis on a domain $D=\{(\evec_{\pvec},\evec_{\vvec},\Xbf_e,\evec_{\omegabf})\in \Real^3\times\Real^3\times\Lcal_{\phi}\times\Real^3\;|\;\|\evec_{\pvec}\|\leq B_{p}\}$, restricting the domain by making $\phi$ and $B_p$ sufficiently small   such that all solutions remain on $D$.

In the following, we consider consider performing the stability analysis on the domain characterized by $D$ with $\phi = 2^{-3}$ with a Lyapunov function candidate $\Lyap = \Lyap^a + \Lyap^p$.

\subsubsection*{Translation error dynamics}
by plugging in the proposed feedback law, the translation error dynamics can be written
\begin{equation}\label{eq:FSFFullerrdynfirst}
m\dot\evec_{\vvec} = m\ddot\pvec - m\ddot\pvec_r = -mg\evec_3 + f\R\evec_3 - 
 m\ddot\pvec_r.
\end{equation}
Note that $\evec_3\R_d\Tr\R\evec_3=\bvec_{d3}\cdot \bvec_3>0$ if we can ensure that~\eqref{eq:firstobs} holds. If so, the term $\tfrac{f}{\evec_3\R_d\Tr\R\evec_3}\R_d\evec_3$ is well defined, and as such,~\eqref{eq:FSFFullerrdynfirst} can be expressed
\begin{align}\label{eq:FSFFull:transdyn}
m\dot\evec_{\vvec}
& = -mg\evec_3- 
 m\ddot\pvec_r  + \frac{f}{\evec_3\R_d\Tr\R\evec_3}\R_d\evec_3 + \bar\fvec,
\end{align}
where
\begin{equation}
\bar\fvec = \frac{f}{\evec_3\R_d\Tr\R\evec_3}\R_d\evec_3 \Big((\evec_3\R_d\Tr\R\evec_3)\R\evec_3 - \R_d\evec_3\Big).
\end{equation}
In addition, recall that the actuating force is computed by
\begin{equation}
\fvec_d = -k_p\evec_{\pvec}-k_v\evec_{\vvec}+mg\evec_3 + m\ddot\pvec_r.
\end{equation}
As $\bvec_{d3}=\fvec_d\|\fvec_d\|^{-1}$, $\fvec_d=\|\fvec_d\|\bvec_{d3} = \|\fvec_d\|\R_d\evec_3$, and
\begin{align}\label{eq:dotsimplificaiton}
\frac{f}{\evec_3\R_d\Tr\R\evec_3}\R_d\evec_3 = 
\frac{\fvec_d\cdot \R\evec_3}{\evec_3\R_d\Tr\R\evec_3}\R_d\evec_3= 
\fvec_d.
\end{align}
Insertion of this expression in~\eqref{eq:FSFFull:transdyn} yields
\begin{align}
m\dot\evec_{\vvec}
 = -k_p\evec_{\pvec}-k_v\evec_{\vvec} + \bar\fvec.
\end{align}
To proceed, we start by bounding $\bar\fvec$ in the control errors using the assumptions of the theorem, but first, we derive the expression for the time-derivative of the part of the Lyapunov function associated with the translation errors.

\subsubsection*{Translation Lyapunov function candidate}
Differentiation of the Lyapunov function associated with the translation dynamics along the solutions of the controlled system yields
\begin{align}\label{eq:FSFFull:lyapderiv}
\dot\Lyap^p 
=& -(k_v-c_p)\|\evec_{\vvec}\|^2 - \frac{c_pk_p}{m}\|\evec_{\pvec}\|^2  - \frac{c_pk_v}{m}(\evec_{\pvec}\cdot \evec_{\vvec})
+ \nonumber\\
&\bar\fvec\cdot \Big(c_pm^{-1}\evec_{\pvec} + \evec_{\vvec}\Big).
\end{align}
Furthermore, by~\eqref{eq:dotsimplificaiton}, we have that
\begin{equation}
\|\bar\fvec\| \leq \|\fvec_d\|\|(\evec_3\R_d\Tr\R\evec_3)\R\evec_3 - \R_d\evec_3\|,
\end{equation}
where the second term is recognized as the sine angle of the eigenaxis rotation angle between $\bvec_{d3}$ and $\bvec_3$, as pointed out in~\cite{lee2010geometric}. As such, we can utilize the fact that this rotation angle is bounded in the control errors on $D$, as per~\eqref{eq:firstobs}. By this simple observation, we obtain
\begin{align*}
\|\bar\fvec\| &\leq \|\fvec_d\|\|(\evec_3\R_d\Tr\R\evec_3)\R\evec_3 - \R_d\evec_3\|\\
&\leq (k_p\|\evec_{\pvec}\| + k_v\|\evec_{\vvec}\| + B)\|(\evec_3\R_d\Tr\R\evec_3)\R\evec_3 - \R_d\evec_3\|\\
&\leq (k_p\|\evec_{\pvec}\| + k_v\|\evec_{\vvec}\| + B) 4\|\evec_{\Xbf}\|\\
&\leq (k_p\|\evec_{\pvec}\| + k_v\|\evec_{\vvec}\| + B) \alpha,
\end{align*}
where the second inequality follows from assumption (A3), and the third and fourth hold for all trajectories on $D$ from the observation regarding the sine angle in~\eqref{eq:secondobs}. Insertion of this bound in~\eqref{eq:FSFFull:lyapderiv} yields
\begin{align}\label{eq:FSFFull:lyapderiv}
\dot\Lyap^p =& -(k_v-c_p)\|\evec_{\vvec}\|^2 - \frac{c_pk_p}{m}\|\evec_{\pvec}\|^2  - \frac{c_pk_v}{m}(\evec_{\pvec}\cdot \evec_{\vvec})
+ \bar\fvec\cdot \Big(\frac{c_p}{m}\evec_{\pvec} + \evec_{\vvec}\Big)\nonumber\\
=& -(k_v(1-\alpha)-c_p)\|\evec_{\vvec}\|^2 - \frac{c_pk_p}{m}(1 - \alpha)\|\evec_{\pvec}\|^2\nonumber\\
&+ \frac{c_pk_v}{m}(1 + \alpha)\|\evec_{\pvec}\|\|\evec_{\vvec}\|\nonumber\\
&+4 \|\evec_{\Xbf}\|\Big(B\Big(\frac{c_p}{m}\|\evec_{\pvec}\| + \|\evec_{\vvec}\|\Big) + k_p\|\evec_{\pvec}\| \|\evec_{\vvec}\|\Big)\nonumber\\
&\leq 
-\zvec_p\Tr \W^{pp}\zvec_p + \zvec_p\Tr\W^{pa} \zvec_a,
\end{align}
with $\W^{pp}$ and $\W^{pa}$ defined as in the proposition statement, in ~\eqref{eq:FSFFullMatrices} and ~\eqref{eq:FSFFullMatrices2}, respectively.
In addition, we note that 
\begin{equation}
\zvec_p\Tr\M^{pp}_1\zvec_p\leq 
\Lyap^p \leq 
\zvec_p\Tr\M^{pp}_2\zvec_p,
\end{equation}
for
$\M_1^{pp}$ and $\M_2^{pp}$ defined as in~\eqref{eq:FSFFullMatrices} of the proposition.

\subsubsection*{Complete Lyapunov function candidate}
In addition to these definitions, take $\M_1^{aa}, \M_2^{aa}, \W^{aa}$, to be the matrices in~\eqref{eq:thm:attitudeSUTcont:domain} associated with the controller in Theorem~\ref{thm:attitudeSUTcont}, also given in~\eqref{eq:FSFFullMatrices}. In addition, define the matrices
\begin{align}
\bar\M_1 &= \mathrm{diag}(\M_1^{pp},\M_1^{aa}),&
\bar\M_2 &= \mathrm{diag}(\M_2^{pp},\M_2^{aa}),&
\bar\W &= \begin{bmatrix}
\W^{pp} & -\frac12\W^{pa} \\
\star & \W^{aa}
\end{bmatrix}.
\end{align}
where $\bar\M_i\succ \Z$ if $\M_i^{jj}\succ \Z$. For the combined Lyapunov function candidate $\Lyap = \Lyap^p + \Lyap^a$, we find that
\begin{equation}
\bar\zvec\Tr\bar\M_1\bar\zvec\leq \Lyap \leq \bar\zvec\Tr\bar\M_2\bar\zvec.
\end{equation}
Differentiating $\Lyap$ along the closed-loop solutions on $D$ yields
\begin{equation}
\dot\Lyap = \dot\Lyap^p + \dot\Lyap^a \leq -\bar\zvec\Tr\bar\W\bar\zvec.
\end{equation}
By assumption (A5), we have that
\begin{align}
\dot\Lyap&\leq -\bar\zvec\Tr\bar\W\bar\zvec\nonumber\\
&\leq -\mineig(\W^{pp})\|\zvec_p\|^2 + \|\W^{pa}\|\|\zvec_p\|\|\zvec_a\|-\mineig(\W^{aa})\|\zvec_a\|^2\nonumber\\
&\leq-\begin{bmatrix}
\|\zvec_p\|\\
\|\zvec_a\|
\end{bmatrix}
\begin{bmatrix}
\mineig(\W^{pp}) & -\frac{1}{2}\|\W^{pa}\|\nonumber\\
\star & \mineig(\W^{aa})
\end{bmatrix}
\begin{bmatrix}
\|\zvec_p\|\\
\|\zvec_a\|
\end{bmatrix}\nonumber\\
&\leq 
-B_{z}(\|\zvec_p\|^2 + \|\zvec_a\|^2).\label{eq:lyaptimederivFSFfull}
\end{align}
Consequently, the Lyapunov function time-derivative is negative definite in $\bar\zvec$ along the solutions of the error dynamics on $D$. By~\eqref{eq:lyaptimederivFSFfull}, it also follows that $\Lyap$ is continuously differentiable on $D$ as all of the signals constituting $\ddot{\Lyap}$ are bounded in the initial errors. As such,~\protect{\cite[Theorem 4.10]{khalil2002nonlinear}} yields that the origin $\bar\zvec = \Z$ is UES on the domain $D$.
\end{proof}
}

\end{document}